\documentclass[11pt,reqno]{amsart}
\usepackage{xfrac}
\usepackage{wrapfig}
\usepackage{adjustbox}

\usepackage[dvipsnames]{xcolor}

\usepackage{tikz}
\usetikzlibrary{patterns} 

\pgfdeclarepatternformonly{wide vertical lines}{\pgfqpoint{-1pt}{-1pt}}{\pgfqpoint{10pt}{10pt}}{\pgfqpoint{2.5pt}{2.5pt}} 
{
  \pgfsetlinewidth{0.4pt}
  \pgfpathmoveto{\pgfqpoint{0pt}{0pt}}
  \pgfpathlineto{\pgfqpoint{0pt}{10pt}}
  \pgfusepath{stroke}
}
\pgfdeclarepatternformonly{wide horizontal lines}{\pgfqpoint{-1pt}{-1pt}}{\pgfqpoint{10pt}{10pt}}{\pgfqpoint{2.5pt}{2.5pt}} 
{
  \pgfsetlinewidth{0.4pt}
  \pgfpathmoveto{\pgfqpoint{0pt}{0pt}}
  \pgfpathlineto{\pgfqpoint{10pt}{0pt}}
  \pgfusepath{stroke}
}

\usepackage{amsthm}
\usepackage[margin=1in]{geometry}
\usepackage[singlespacing]{setspace}
\usepackage{amsmath,amsfonts, amssymb}

\usepackage{enumerate}
\usepackage{enumitem}
\usepackage{subfigure}
\usepackage{setspace}
\usepackage[ruled]{algorithm2e} 

\usepackage{etoolbox}
\usepackage{caption}
\usepackage{bbm}
\PassOptionsToPackage{hyphens}{url}

\usepackage{hyperref}

\hypersetup{
     colorlinks=true,
     linkcolor=blue,
     filecolor=blue,
     citecolor = blue,      
     urlcolor=cyan,
     }

\usepackage[authoryear,sort]{natbib}
 \bibpunct[, ]{(}{)}{,}{a}{}{,}%

\usepackage[margin=1in]{geometry}
\usepackage[normalem]{ulem}

\usepackage{graphicx}
\usepackage{color}

\newtoggle{version}
\toggletrue{version}

\newtoggle{comments}
\toggletrue{comments}

\newtheorem{definition}{\sc Definition}  [section]
\newtheorem{lemma}{\sc Lemma}  [section]
\newtheorem{claim}{\sc Claim}  [section]
\newtheorem{theorem}{\sc Theorem}  
\newtheorem*{theorem*}{\sc Theorem}
\newtheorem{corollary}{\sc Corollary} [section]
\newtheorem{proposition}{\sc Proposition}[section]
\newtheorem{remark}{\sc Remark}[section]

\newcommand{\threshold}{\mathcal{I}^{\beta,\epsilon}}

\newcommand{\scl}{s}

\usepackage{amsmath}
\usepackage{amsthm}
\usepackage{graphicx}
\usepackage{tabularx}
\usepackage{comment}
\usepackage{xcolor}
\usepackage{url}
\usepackage{tikz,pgfplots}
\usetikzlibrary{shapes.geometric, arrows}
\usepackage{thm-restate}
\usepackage{natbib}

\newcolumntype{C}{>{\centering\arraybackslash}p{1em}}

\newcommand{\Xomit}[1]{}

\newcommand{\randbud}[1]{\mathcal{I}_{#1}}
\newcommand{\z}{{\bf z}}
\newcommand{\y}{{\bf y}}
\newcommand{\x}{{\bf x}}
\newcommand{\prices}{{\bf p}}

\newcommand{\Set}{V}
\newcommand{\V}{V}
\newcommand{\A}{A}

\SetKwInput{KwData}{Input}
\SetKwInput{KwResult}{Output}

\def\final{0}  
\def\iflong{\iffalse}
\ifnum\final=0  
\newcommand{\thanh}[1]{{\color{blue}[{\small Thanh: \bf #1}]\marginpar{\color{red}*}}}
\newcommand{\ysl}[1]{{\color{brown}[{\small Young-San: \bf #1}]\marginpar{\color{red}*}}}
\newcommand{\haoyu}[1]{{\color{olive}[{\small Haoyu: \bf #1}]\marginpar{\color{red}*}}}
\else 
\newcommand{\thanh}[1]{}
\newcommand{\ysl}[1]{}
\newcommand{\haoyu}[1]{}
\fi  

\title{A few good choices}

\author{Haoyu Song}
\thanks{Department of Computer Science, Purdue University; \href{mailto:song522@purdue.edu}{song522@purdue.edu}}

\author{Th\`{a}nh Nguyen}
\thanks{Daniels School of Business, Purdue University; \href{mailto:nguye161@purdue.edu}{nguye161@purdue.edu}}

\author{Young-San Lin}
\thanks{
Melbourne Business School; \href{mailto:y.lin@mbs.edu}{y.lin@mbs.edu}}

\date{\today}

\begin{document}

\onehalfspacing

\begin{abstract}

A \emph{Condorcet winning set} addresses the Condorcet paradox by selecting a few candidates—rather than a single winner—such that no unselected alternative is preferred to all of them by a majority of voters. This idea extends to \emph{$\alpha$-undominated sets}, which ensure the same property for any $\alpha$-fraction of voters and are guaranteed to exist in constant size for any $\alpha$. However, the requirement that an outsider be preferred to \emph{every} member of the set can be overly restrictive and difficult to justify in many applications. Motivated by this, we introduce a more flexible notion: \emph{$(t, \alpha)$-undominated sets}. Here, each voter compares an outsider to their $t$-th most preferred member of the set, and the set is undominated if no outsider is preferred by more than an $\alpha$-fraction of voters. This framework subsumes prior definitions, recovering Condorcet winning sets when $(t = 1, \alpha = 1/2)$ and $\alpha$-undominated sets when $t = 1$, and introduces a new, tunable notion of collective acceptability for $t > 1$. We establish three main results:
\begin{itemize}
    \item We prove that a \((t, \alpha)\)-undominated set of size \(O(t/\alpha)\) exists for all values of \(t\) and \(\alpha\).
    
    \item We show that as \(t\) becomes large, the minimum size of such a set approaches \(t/\alpha\), which is asymptotically optimal.
    
    \item In the special case \(t = 1\), we improve the bound on the size of an \(\alpha\)-undominated set given by \cite{sixCandidates}. As a consequence,  we show that a Condorcet winning set of five candidates exists, improving their bound of six.
\end{itemize}

\end{abstract}

\maketitle

\section{Introduction}



Computational social choice is an emerging area at the intersection of economics and computer science, with broad applications ranging from voting and public goods provision to fairness in machine learning and recent advances in AI-human alignment. At its core is the fundamental problem of aggregating heterogeneous preferences in a fair and principled way. Yet, the generality of the social choice framework  often leads to strong impossibility results, which necessitate the development of compromise solutions or approximation approaches.

A common strategy in response to such challenges is to allow for the selection of multiple candidates, rather than a single winner, to better accommodate conflicting preferences and ensure broader acceptability. For example, in contexts like prize distribution, awarding a group of candidates  may better reflect the preferences of the  judging panel when there is no clear front-runner. Similarly, in committee selection, slightly increasing the committee size can lead to more effective and representative outcomes. Even when only one or a few candidates must ultimately be chosen, presenting a shortlist for further evaluation is common and can improve the decision-making process.

The idea of selecting more than one candidate has deep roots in voting theory, motivated in large part by impossibility results that emerge when restricting attention to a single winner. A prominent example is the Condorcet paradox \citep{de1785essai}, which illustrates that majority preferences can be cyclical, preventing any single candidate from defeating all others in pairwise comparisons. More strikingly, generalized versions reveal that in some voting instances, no single candidate can even secure support from a positive fraction of voters.
Despite these challenges, producing a voting outcome that garners support from a sufficiently large fraction of the electorate remains a central objective. One approach to achieving this is to relax the single-winner constraint and instead allow for the selection of multiple candidates.

One such relaxation extends the notion of a Condorcet winner to a Condorcet winning set (committee), as introduced by \citet{elkind2015condorcet}. A Condorcet winning committee is a set of candidates such that no candidate outside the set is preferred to \emph{all} members of the set by a majority of voters. Identifying the minimal number $k$ for which a Condorcet winning committee of size $k$ always exists across all elections remains an open problem. The current best-known upper bound is 6 established by \cite{sixCandidates}, while the lower bound is $3$.

The Condorcet winning set, as described above, is not the only way to define outcomes with majority support. It boils down to how the comparison between an alternative and a set of alternatives are defined. For example, \cite{mcgarvey1953theorem} studies an alternative definition, requiring that for every candidate outside the selected set, there exists at least one candidate within the set that gains majority support over the external candidate. This is a stronger requirement, as it demands that all voters compare the outside candidate to a single member of the committee. However, \cite{mcgarvey1953theorem} shows that such a committee of constant size does not necessarily exist. \cite{charikar2025approximately,bourneuf2025dense} show that if the majority condition is relaxed to $\frac{1}{2} - \epsilon$ for any positive $\epsilon$, a constant-size committee will always exist. However, a drawback of this stronger requirement is that it typically leads to very large bounds on committee size for small $\epsilon$, limiting practical applicability.\footnote{For instance, with $\epsilon = 0.01$, \cite{bourneuf2025dense} gives a bound of 5{,}227{,}032, while \cite{charikar2025approximately} yields a bound of 3{,}927.}

A more serious limitation of both concepts discussed above is their reliance on the assumption that each voter compares the outside alternative with only a single member of the selected committee. This assumption becomes especially problematic in the context of large committees. For example, if an alternative is excluded from a committee of 20 members, it is not well justified to attribute this outcome to a preference for just one selected candidate. More plausibly, the majority supports a combination of multiple committee members over the outside alternative. Ignoring this aggregated preference risks misrepresenting the collective rationale behind the committee’s composition.  A particularly relevant application is the selection of papers for a conference or research proposals for funding, where decisions are based on diverse and sometimes conflicting expert evaluations. In these settings, the selected set of candidates corresponds to the accepted papers or funded projects. To justify the rejection of a given submission, it is not sufficient to argue that it is weaker than the single best entry. Rather, the decision should be evaluated in the context of the entire selected set. For example, one may want to assess whether the rejected paper had less support than the median accepted submission, or little support compared with those in the top quartile.





\subsection*{Our Contributions}

The goal of this paper is to introduce a general framework for analyzing Condorcet-style winning sets—one that allows voters to compare a single alternative against a set of selected candidates. This framework involves two key generalizations. First, it extends the notion of how a voter may prefer a committee over an alternative outside the committee. Second, it generalizes the concept of a \emph{majority} to any fixed fraction of the electorate. Specifically, we define a \emph{$(t, \alpha)$-undominated set} as a subset $C$ of candidates with $|C| \ge t$, such that for every candidate $a \notin C$, at least a $1 - \alpha$ fraction of voters prefer at least $t$ members of $C$ over $a$. The notion of $(t, \alpha)$-undominance captures both key generalizations. The central question we study is:
\begin{center}
    {\it What is the minimal size of a $(t, \alpha)$-undominated set that exists in all voting instances?}
\end{center}

A special case arises when $t = 1$, in which case a $(1, \alpha)$-undominated set is simply referred to as an \emph{$\alpha$-undominated set}, introduced in \cite{elkind2015condorcet,elkind2011choosing}. 
When $t = 1$ and $\alpha = 1/2$, the concept coincides with the classical Condorcet winning set. 
To the best of our knowledge, the case of $t > 1$ has not been previously explored.

\Xomit{
Generally speaking, the central question we investigate is:
\begin{center}
    {\it What is the minimum size for a generalized Condorcet winning set in any voting instance?}
\end{center}

To generalize \emph{majority}, it is natural to consider a parameter $\alpha \in [0,1]$ served as a threshold for the proportion of voters who deviate from the selected committee to a single alternative. A committee is \emph{$\alpha$-undominated} if for any fixed outside alternative, the fraction of voters who prefer the outside alternative to the top choice in the committee is at most $\alpha$. A $\frac{1}{2}$-undominated set is equivalent to a Condorcet winning set. This notion is helpful to tackle the following question raised by \cite{elkind2015condorcet}: 
\begin{center}
\textit{Given a target size $k$, what is the smallest $\alpha$ for which there exists a set of size $k$ which is $\alpha$-undominated for any election instance $\mathcal{E}$?}
\end{center}
}

In this paper, we establish three main results. We begin with the case $t = 1$ and prove the following theorem in Section~\ref{sec:condorcet}.

\begin{restatable}{theorem}{thmundom} \label{thm:undominate}
    Given $k \in \mathbb{N}$ and $\beta \in [0,1]$, a $(\beta + (1 - \beta)^k)$-undominated committee of size $k$ exists.
\end{restatable}

This theorem improves the undominance ratio in \cite{sixCandidates} for all \(k\), establishes the existence of a Condorcet winning set of size 5, and tightens the bound of 6 in that paper.

We study the case $t \ge 2$ in Section~\ref{sec:generalized_cond} and prove the following result.

\begin{restatable}{theorem}{thmgeneralupperbound} \label{thm:generalized-cond}
Given $\alpha \in (0,1]$, there exists a set of size $\lfloor \delta(t)\cdot \frac{t}{\alpha} \rfloor$ which is $(t,\alpha)$-undominated. Here, $1<\delta(t)\le 4.75$ for all integral $t\ge 2$  and $\delta(t) \rightarrow 1$ as $t\rightarrow \infty$.
\end{restatable}

In Section~\ref{sec:generalized_cond}, we provide the precise formulation of $\delta(t)$ and compute its values for small values of $t$. For instance, we find that $\delta(2) = 4.75$, $\delta(3)=4.11$ and $\delta(4)=3.69$. Overall, the result implies that selecting $O(t/\alpha)$ candidates suffices to guarantee a $(t,\alpha)$-undominated committee. 

We also establish the following lower bound to show the asymptotic tightness of this guarantee.

\begin{restatable}{theorem}{thmgenerallowerbound} \label{thm:k-alpha-lower}
If $\alpha < \frac{t+1}{k+1}$, then there exist elections that admit no $(t,\alpha)$-undominated committee of size $k$.
\end{restatable}

Theorem~\ref{thm:k-alpha-lower} states that any \((t, \alpha)\)-undominated committee must contain at least \(\left\lceil \frac{t+1}{\alpha} - 1 \right\rceil\) candidates. Combined with Theorem~\ref{thm:generalized-cond}, this yields an asymptotically tight bound of \(\frac{t}{\alpha}\) for the minimum committee size.

This asymptotic bound provides an interpretation of fairness in group selection and highlights a sharp trade-off. Specifically, a candidate can be legitimately excluded only if fewer than an \( \alpha \)-fraction of voters rank them above the committee’s most-preferred \( \alpha \)-fraction. In other words, exclusion is justified when the candidate fails to outperform the committee’s top tier in the eyes of a sufficiently large portion of the electorate.

A particularly transparent case arises when \(\alpha = 1/2\): a candidate is excluded only if a majority of voters prefer the median member of the committee over that candidate. More generally, the parameter \(\alpha\) serves a dual role---governing both the benchmark for comparison and the threshold of support---highlighting a delicate balance between depth of preference and breadth of agreement. Crucially, the \(\frac{t}{\alpha}\) bound is not only intuitive but also asymptotically optimal; no smaller committee can satisfy this level of undominance.

\subsection*{Technical Overview}

The technical contribution of our paper lies in adapting the Lindahl Equilibrium with Ordinal preferences (LEO), recently introduced by \cite{nguyen-song}, to the single-winner setting. In our formulation, LEO consists of a continuous distribution of income (or tokens) assigned to each voter, along with personalized prices for each voter-candidate pair. Voters consume probabilistically—that is, they receive a lottery over candidates determined by their income distribution and individual prices. A centralized agent then selects a lottery over candidates that maximizes total revenue.
A condition analogous to market clearing must be satisfied: the consumption of each voter must closely match the selected lottery. Discrepancies are allowed only in cases where the price of a candidate is zero.

This framework is inspired by  the classical Lindahl equilibrium, where voters have convex preferences over lotteries, receive a fixed endowment of tokens (normalized to one), and choose lotteries based on personalized prices. The key difference is to replace deterministic incomes with randomized ones, which allows us to accommodate purely ordinal preferences without relying on a cardinal utility representation or convex relaxation.

To construct a Condorcet winning set within this framework, we introduce a \emph{threshold distribution}, which is a continuous approximation of a discrete distribution placing probability \(\beta\) on 0 and \(1 - \beta\) on 1. This construction enables us to show that if a candidate falls within a voter's top \(1 - \beta\) percentile of the lottery constructed from LEO, the price they face for that candidate is close to 1. We leverage this fact and the property of a LEO to bound the number of voters who place an “outside” candidate within their top \(1 - \beta\) percentile. To construct a Condorcet winning set, we carefully choose the parameter \(\beta\) and then randomly select \(k\) candidates from the support of the LEO outcome.

Compared to \cite{sixCandidates}, our approach not only yields a stronger result but is also simpler and more amenable to generalization—particularly to the construction of \((t, \alpha)\)-undominated sets. Specifically, in the generalized setting—where comparisons between the selected set and an outside candidate depend on the top \(t\) members of the set—we modify LEO into a scaled version, which we call SLEO (Scaled Lindahl Equilibrium with Ordinal preferences). The key motivation for this modification is to ensure that each voter's consumption is spread across \(O(t)\) candidates, rather than concentrating entirely on their single most preferred option. This diffusion allows us to bound the price of any outside candidate that is ranked above one of the top \(t\) candidates in the selected set.

However, this modification introduces a tradeoff: the SLEO outcome is no longer a lottery over committees but a fractional allocation across individual candidates, with total weight at most equal to the committee size. To handle this, we take two additional technical steps:  
(a) we convert the fractional solution into a randomized committee of the desired size using dependent rounding with negative correlation, and  
(b) we implement a more refined, adaptive iterative procedure to construct the final committee while preserving the desired properties.

\subsection*{Related Literature}





This paper contributes to the literature on (computational) social choice by addressing several fundamental questions related to committee selection and preference aggregation.

\subsubsection*{Condorcet winning set} One of the most well-known properties in this literature is the Condorcet property, which requires that the chosen outcome be preferred by a majority over any other alternative. However, such an outcome is often unattainable \citep{de1785essai}.
 To address this, the literature restricts preference domains or adopts weaker axioms, such as {Condorcet consistency}, which selects a Condorcet winner when one exists\footnote{Randomization is also widely studied to resolve the Condorcet paradox; see \cite{brandt2017rolling} for a survey.}. We follow \cite{elkind2015condorcet} to extend the set of possible outcomes to sets of candidates. \cite{elkind2015condorcet} show that  a {Condorcet winning set} of size at most the logarithmic of the number of candidates  always exists. A constant-size Condorcet winning set (of size 32) was implicitly established in \cite{ApproStable}, recently improved to 6 by \cite{sixCandidates}. Our results further reduce this bound to 5, although a gap remains between this and the lower bound of~3.

\subsubsection*{Undominated and dominating set} The notion of majority support naturally extends to any \(\alpha\) fraction, giving rise to the concept of an \(\alpha\)-undominated set, as studied by~\cite{sixCandidates}. In our framework, this corresponds to the special case \( t = 1 \). Even in this case, our approach yields tighter bounds for all \(\alpha\) compared with prior work.

As discussed in the introduction, a stronger notion is that of an \(\alpha\)-dominating set: a set \( C \) is \(\alpha\)-dominating if, for every \( a \notin C \), there exists \( b \in C \) such that at least an \(\alpha\) fraction of voters prefer \( b \) to \( a \). For \(\alpha \geq \frac{1}{2}\), the smallest such set can be arbitrarily large~\cite{mcgarvey1953theorem}. In contrast, recent work~\cite{bourneuf2025dense, charikar2025approximately} shows that for any \(\alpha < \frac{1}{2}\), constant-size \(\alpha\)-dominating sets exist, with size depending on \(\epsilon = \frac{1}{2} - \alpha\).

Our notion of a \((t, \alpha)\)-undominated set is incomparable to \(\alpha\)-domination. It is stronger in requiring each outside candidate to be compared against the top \( t \) candidates in the set, but weaker in allowing each voter to compare with a different subset of \( t \) candidates. A key advantage of our notion is that a constant-size \((t, \alpha)\)-undominated set exists for any fixed \( t \) and \( \alpha \in (0,1)\).

\subsubsection*{Other concepts} Several alternatives to Condorcet winners have been proposed, though many face similar limitations: non-existence or the need to select a large number of candidates, as in $\alpha$-dominating sets for $\alpha \geq \frac{1}{2}$. For example, \cite{fishburn1981majority} introduced the idea of a \emph{majority committee}—a set preferred by a majority over any other set of the same size—though this requires preferences over committees and may not always exist. The Smith set~\citep{smith1973aggregation}, a minimal set beating all outsiders in pairwise comparisons, offers a stronger form of the Condorcet principle. The \emph{bipartisan set}~\citep{laffond1993bipartisan}, defined via maximal lotteries, is a subset of the Smith set and captures randomized generalizations. We also refer the reader to~\cite{brandt2016handbook} for a comprehensive survey of related concepts based on tournament graph structures.

\subsubsection*{Committee Selection with Ranking Preference}

The question of relaxing the choice set to allow the selection of multiple candidates is closely related to the settings of committee selection and, more broadly, participatory budgeting. In particular, this connection is captured through the concept of the \emph{core}. A key difference is that in participatory budgeting, the core is typically more demanding, as it permits deviations by sets of candidates rather than by individual candidates. Defining such a core requires extending preferences to sets of candidates, which often leads to weaker approximation guarantees compared to the Condorcet-style conditions we consider. For example,~\cite{ApproStable} establish a 16-approximation guarantee for  committee selection with ranking-based preferences, which is later improved to 9.8217 by \cite{sixCandidates}  and to 5.03 by \cite{nguyen-song}. However, this only implies a bound of 11 on the size of a Condorcet-winning set—significantly looser than the bound of 5 that we establish in this paper. More generally, in terms of $\alpha$-undominance set, our approach can also be shown to generate tighter bound for any $\alpha\geq 0.093.$

Our techniques are related to those in~\cite{nguyen-song}, who introduce the Lindahl Equilibrium with Ordinal preferences (LEO) to study general monotonic participatory budgeting. In contrast, our single-winner setting allows for tighter bounds, enabled by a more specialized and refined analysis. Moreover, as we show in this paper, the technique in \cite{nguyen-song} does not directly yield a $(t,\alpha)$-undominated set in our setting and requires  modifications to the concept of LEO to be applicable.

\subsubsection*{Intransitive dice}
The connection between voting theory and the Condorcet paradox has a surprising parallel in the phenomenon of \emph{intransitive dice}, as observed in~\cite{sixCandidates, charikar2025approximately}. In such a setup, each die has a higher chance of beating the next in a cycle—yet paradoxically, the last die outperforms the first.
This curious behavior, popularized by~\cite{gardner1970paradox}, has inspired a substantial body of mathematical research. For a deeper exploration of this link, see~\cite{charikar2025approximately}. Uncovering the full extent of its relevance to our setting presents an intriguing avenue for future investigation.

\subsection*{Organization}

In Section \ref{sec:pre}, we introduce necessary notations and our main technical tool, Lindahl Equilibrium with Ordinal Preferences (LEO). In Section \ref{sec:condorcet}, we prove that five candidates suffice for a Condorcet winning set. In Section \ref{sec:generalized_cond}, we prove that $O(t/\alpha)$ candidates are sufficient and required to ensure a $(t,\alpha)$-undominated set.

\section{Notations and Preliminaries} \label{sec:pre}
Let \( \V = \{1, 2, \ldots, n\} \) denote the set of voters (agents), and let \( A \) represent the set of \( m \) alternatives (or candidates).  Each agent \( v \in \V \) has a strict preference relation \( \succ_v \) over the alternatives in \( A \). \( a \succeq_v b \) means that either $a=b$ or agent \( v \) prefers \( a \) to \( b \). An election instance is denoted by $\mathcal{E}=(V,A,\succ)$.


Throughout the paper, we use bold lowercase letters such as $\x, \y$ to denote a vector and regular lowercase letters such as $x_{i,j}, y_k$ to denote a coordinate of a vector.

\subsection*{Lindahl Equilibrium with Ordinal Preferences} \label{sec:prelim-leo}

We adapt the concept of Lindahl Equilibrium under Ordinal Preferences (LEO), introduced in \cite{nguyen-song}, to our setting. As discussed in the introduction, LEO itself is an ordinal adaptation of the classical Lindahl equilibrium \citep{Lindahl70}.


We assume that each voter selects only individual candidates or an outside option, denoted by $\emptyset$.  Recall that each $v \in V$ is associated with a strict preference $\succ_v$ over $A$. We extend the use of $\succ_v$ over $A \cup \{\emptyset\}$ by having $\emptyset$ as the least preferred, while the ranking over $A$ remains the same. 

Each voter $v$ is endowed with a random income  $\randbud{v}$ supported on $[0,1]$  and has a  \textit{personalized} price vector $\prices_v \in \mathbb{R}_{+}^{|A \cup \{\emptyset\}|}$. We use $p_{v,a}$ to denote the personalized price of $a\in A \cup \{\emptyset\}$ for voter $v$ with $p_{v,\emptyset} = 0$. We assume that voters share a common income distribution and therefore denote it as $\randbud{}$. 


Under a fixed income $b \in \mathbb{R}_+$, the voter $v$'s demand is $ \max_{\succ_{v}} \{a: a \in A \cup \{\emptyset\} \text{ and } p_{v,a} \leq b  \}$, which is her most preferred and affordable alternative within her income. Notice that when the prices of all candidates are greater than $b$, the demand will be the outside option, $\emptyset$. Given the random income distribution  $\randbud{}$, voter $v$'s \textit{random demand} is defined as
$$
\mathcal{D}_{v}(\prices_v, \randbud{}) := \left\{\max_{\succ_{v}} \{a: a \in A \cup \{\emptyset\} \text{ and } p_{v,a} \leq b \} \mid b \sim \randbud{} \right\}.
$$  


Thus, under the assumptions that the random income distribution is strictly positive with probability 1 and that $p_{v,\emptyset} = 0$, the resulting random demand is non-empty and constitutes a lottery on $A \cup \{\emptyset\}$.

Lindahl equilibrium also involves a centralized agent called the \emph{producer}, who decides on the quantity of each candidate in $A$ to produce. From the producer's perspective, there is a unit cost associated with each candidate $a \in A$. The producer has a total budget of $B >0$ and can only produce the quantity $\z\in  \mathbb{R}_+^{|A|} $ satisfying $\mathbf{1}^T \z = B$. Thus, the budget $B$ determines the size of the committee. The producer aims to maximize total revenue. Therefore, the set of optimal strategies for the producer is 
$$
\mathcal{P}(B,\prices):= \arg\max_{z\in  \mathbb{R}_+^{|A|}} \sum_{a\in A} \left(\sum_{v\in \V}p_{v,a}\right)\cdot z_a  \text{ subject to } \mathbf{1}^T \z = B.
$$

We have the following definition of the Lindahl equilibrium with ordinal preference (LEO) for our setting adapted from \cite{nguyen-song}. 


\begin{definition} \label{def:LEO}
Given a random income $\randbud{}$ supported on $[0,1]$ and a total budget $B>0$, the Lindahl equilibrium with ordinal preference (LEO) $(\x, \y, \prices)$ consists of personalized prices $\prices_{v} \in [0,1]^{|A \cup \{\emptyset\}|}$ with  $p_{v,\emptyset} = 0$ and individual consumptions $\x_v \in [0,1]^{|A \cup \{\emptyset\}|}$ for each voter $v \in V$, and a common allocation $\y \in [0,B]^{|A|}$ such that the following holds: 
\begin{enumerate}[label=(\arabic*)]
    \item \label{leo-1} for each $a\in A \cup \{\emptyset\}$, $x_{v,a}=  \Pr[\mathcal{D}_{v}(\prices_v, \randbud{}) =a]$; 
    \item \label{leo-2} for each $a\in A$, $x_{v,a} \le  y_a $ and  if the strict inequality holds then  $p_{v,a} = 0$; 
    \item \label{leo-3} $\y\in \mathcal{P}(B,p)$.
\end{enumerate}
\end{definition}

Compared with the traditional Lindahl equilibrium with convex and continuous preferences, there are two key differences in our setting. In the classical Lindahl equilibrium, each voter has a fixed, constant income or tokens, that is normalized to 1, and there is a common allocation \( \y \) along with a set of individual prices such that: (i) the common allocation is the optimal allocation for each voter given their individual prices, and (ii) \( \y \) is also the revenue-maximizing allocation for the producer.

In our formulation, the first key departure from the classical setting is that each voter's income is modeled not as a fixed constant but as a random variable with a continuous distribution. This modification is essential because we work with ordinal preferences over discrete choices; introducing income uncertainty ensures that each voter's resulting random demand varies continuously with prices, without requiring an extension of preferences from discrete choices to lotteries. Second, we relax the requirement that each individual's demand must exactly match the common allocation. Instead, we adopt a condition analogous to market clearing, formalized as condition~\ref{leo-2} in Definition \ref{def:LEO}. 

Together, these two modifications are essential for applying Kakutani's fixed-point theorem and establishing the existence of a LEO. For completeness, the full proof is provided in Appendix~\ref{app:LEO}.


\begin{restatable}{theorem}{thmleo}(\cite{nguyen-song}) \label{thm:existence}
    Let $\randbud{}$ be a random income with support in the unit interval and its cumulative distribution function is continuous, then there exists a LEO.
\end{restatable}

\section{Five Candidates Suffice to Win a Voter Majority}
\label{sec:condorcet}

We start with some necessary notations. Given an election $\mathcal{E}=(V,A,\succ)$, for \(v \in V\),\( C \subseteq A \), and \( a \in A \setminus C \), we write \( a \succ_v C \) to indicate that \( v \) strictly prefers \( a \) to all alternatives in \( C \), i.e., \( a \succ_v c \) for all \( c \in C \). A Condorcet winning set is a set of candidates $C \subseteq A$ such that for every candidate $a \in A \setminus C$, the minority of the voters prefer $a$ to any candidates in $C$, i.e., $\lvert\{v \in V \mid a\succ_v C\}\rvert \le \lfloor \frac{n}{2} \rfloor$ for all $a \in A$.
A generalization of the Condorcet winning set is the notion of \emph{$\alpha$-undominated set}.

\begin{definition}
A subset $C \subseteq A$ is an \emph{$\alpha$-undominated set} if for every candidate $a \in A \setminus C$, 
$$\lvert\{v \in V : a \succ_v C \}\rvert \le  \lfloor \alpha n \rfloor .$$ 
\end{definition}

We note that a $\frac{1}{2}$-undominated set is equivalent to a Condorcet winning set.  
Our main result in this section is  the following theorem.



\thmundom*

Given $k \in \mathbb{N}$, let $\alpha(k)$ be the undominance ratio that we aim to minimize. From Theorem \ref{thm:undominate}, we have $\alpha(k):= \inf_{\beta \in [0,1]}\{\beta + (1-\beta)^k\}$. We compare our $\alpha(k)$ with the one from \cite{sixCandidates} Theorem 5 in Table \ref{table:und-ratio}. 

\begin{table}[!htb]
\begin{center}
    \begin{tabular}{|c|c|c|}
        \hline
        $k$ & $\alpha(k)$ from Theorem \ref{thm:undominate} & $\alpha(k)$ from \cite{sixCandidates} Theorem 5 \\
        \hline
        $2$ & 0.75 & 0.798134 \\
        \hline
        $3$ & 0.615100 & 0.673795 \\
        \hline
        $4$ & 0.527530 & 0.588554 \\
        \hline
        $5$ & 0.465008 & 0.525719 \\
        \hline
        $6$ & 0.417645 & 0.477066 \\
        \hline
        $7$ & 0.380269 & 0.438041 \\
        \hline
        $8$ & 0.349879 & 0.405896 \\
        \hline
    \end{tabular}
\end{center}
\caption{Bounds on the undominance ratio for size $k$.} \label{table:und-ratio}
\end{table}

Moreover, we obtain the following result as a direct consequence of Theorem~\ref{thm:undominate}.
\begin{corollary}
There exists a 5-candidate Condorcet winning set for any election.
\end{corollary}

To prove Theorem~\ref{thm:undominate}, we use our LEO with a total budget $B = 1$ and a specially designed income distribution, which we call the \emph{threshold distribution}. The threshold distribution $\threshold$ has two parameters, $\beta$ and $\epsilon$, and serves as a smoothed approximation of a Bernoulli distribution that assigns income $1$ with probability $\beta$ and $0$ with probability $1 - \beta$. The small parameter $\epsilon$ ensures that the distribution remains continuous.

\begin{definition}
    Given \(\beta, \varepsilon\in(0,1)\), the \emph{threshold distribution} \(\threshold\) satisfies the following properties.
    \begin{enumerate}
        \item The random variable $X \sim \threshold$ has a continuous CDF $F(x) = \Pr[X \le x]$.
        \item $F(1-\beta) = 1 - \varepsilon$, i.e., $\Pr[X \ge 1 - \varepsilon] = \beta$.
        \item The expected value of $X$ satisfies $\mathbb{E}[X] \le \beta$, that is, $\int_{0}^1F(x)dx\le \beta$.
    \end{enumerate}
\end{definition}
\begin{figure}[htbp] 
    \begin{center}
        \includegraphics[width=0.4\linewidth]{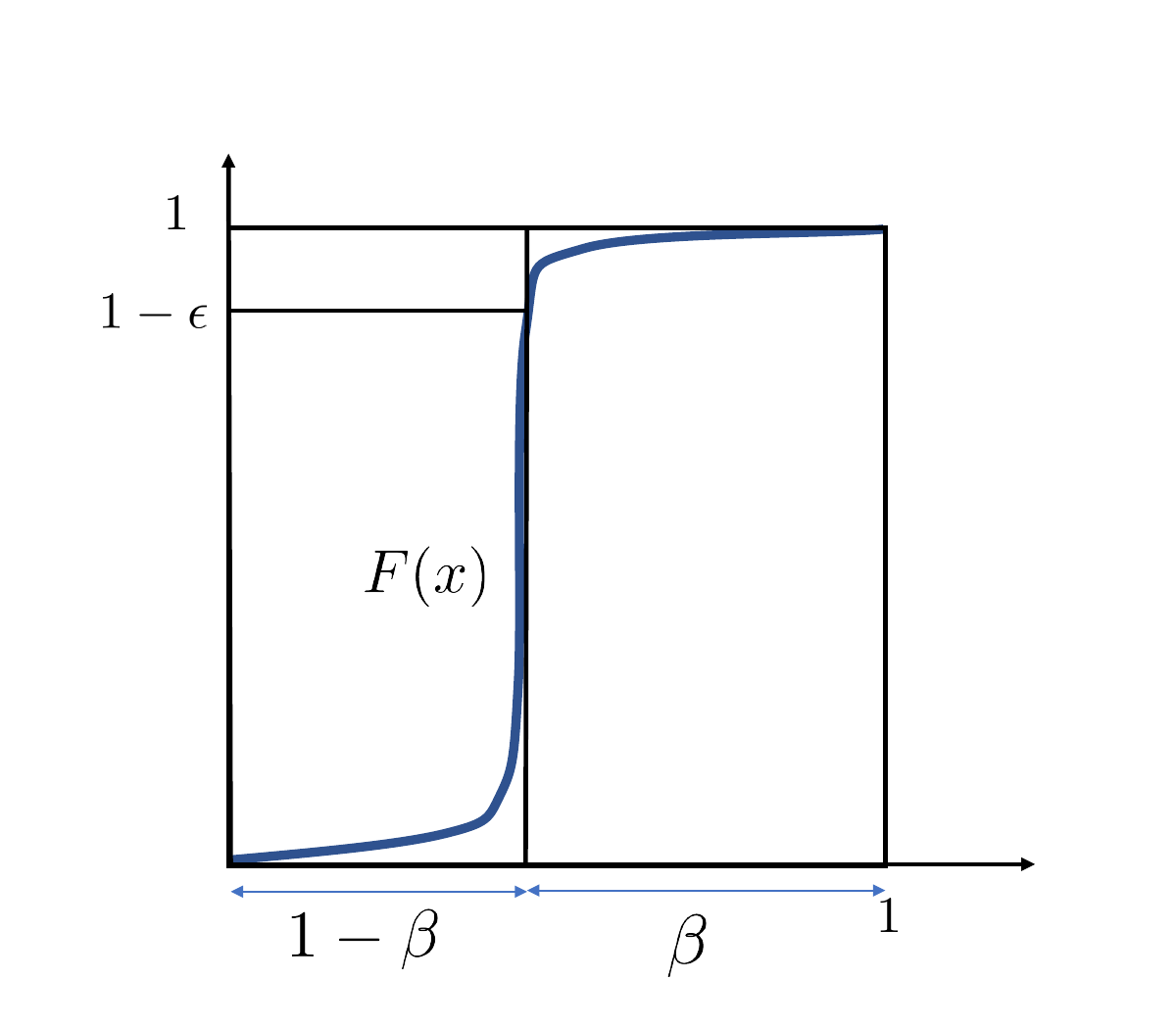}
    \captionof{figure}{Threshold Income distribution}
    \label{fig:001}
    \end{center}
\end{figure}
It is straightforward to construct \(\threshold\) for any \( \beta, \varepsilon \in (0,1) \). Consider the distribution \( \mathcal{I}^{\beta,0} \), which assigns probability \( \beta \) to 1 and \( 1 - \beta \) to 0. While this distribution satisfies the last two properties, it is not continuous. However, a slight perturbation by \( \varepsilon \) ensures continuity (see Figure~\ref{fig:001}).

We first show that  when $B=1$, the common allocation $\y$ of a LEO coincides will all individual random demand and forms a lottery over the alternatives.

\begin{claim} \label{cl:lottery}
    If  $(\x,\y,\prices)$ is a LEO with $B=1$, then  $\sum_{a \in A} y_a = 1$ and $x_{v,a} = y_a$ for all $a\in A, v \in V$.
\end{claim}
\begin{proof}
First, by definition, \(\sum_{a \in A} y_a = B = 1\). Second, suppose \(x_{v,a} < y_a\). Then, by condition~\ref{leo-2} of the LEO, it must be that \(p_{v,a} = 0\). Moreover, \(\sum_{a \in A} x_{v,a} < 1\), which implies that voter \(v\) consumes \(\emptyset\) with positive probability. This contradicts the definition of random demand for \(v\), because if \(p_{v,a} = 0\), then \(v\) should not have chosen \(\emptyset\), since \(a\) was free and strictly preferred.
\end{proof}

To prove Theorem~\ref{thm:undominate}, we consider a random committee formed by independently sampling $k$ alternatives from the lottery induced by the LEO mechanism with total budget $B = 1$ and random income drawn from the threshold distribution $\threshold$. We then show, via an expectation argument, that there exists at least one realization of this random process that satisfies the conditions of Theorem~\ref{thm:undominate}.

The argument bounding the number of voters who may prefer an outside candidate proceeds in two parts. First, we define a notion of \emph{coverage}: intuitively a voter is said to be covered by the committee if at least one of its members lies within the top $\beta$ percentile of the voter’s preferences under the LEO lottery. The probability that a voter is covered can be bounded based on $\beta$, and by applying a standard averaging argument, we obtain a deterministic committee that fails to cover only a small fraction of voters. This yields the first bound.

Second, for the voters who are covered, we show that only a limited number can strictly prefer some outside candidate over the entire committee. If too many did, the total price they assign to that candidate would exceed the total available income. However, a key property of the LEO is that the total price any candidate can receive across all voters is bounded above by the total income, which is at most $\beta n$. This provides the second bound. The sum of the two bounds then gives the desired result.

We introduce the notion of \emph{boundary candidates}.

\begin{definition} \label{def:boundary}
    Given a LEO $(\x,\y,\prices)$, for each voter $v \in \V$, the \emph{boundary candidate} $a_v$ is the most preferred candidate with price at most $1 - \varepsilon$, that is, $a_v := \max_{\succ_v} \{a \in A\cup \{\emptyset\} \mid p_{a,v} \le 1 - \varepsilon\}$.
\end{definition}

The following claim is a direct  consequence  of this definition.  It  states that the price of a candidate that is more preferred than the boundary candidate is high.

\begin{claim} \label{cl:high-price}
    Given a LEO $(\x,\y,\prices)$, for any \(v \in V\) and $a \in A$ such that $a \succ_v a_v$, \(p_{v,a} > 1 - \varepsilon\).
\end{claim}

\begin{proof}
    For the sake of contradiction, suppose $p_{a,v} \le 1 - \varepsilon$. Then, by Definition \ref{def:boundary}, we have that $a \preceq_v a_v$, contradicting $a \succ_v a_v$.
\end{proof}

Now we are ready to define \emph{covered} and \emph{uncovered} voters.

\begin{definition}Given a LEO $(\x,\y,\prices)$ and a set of candidates $C \subseteq \A$, a voter $v \in V$ is \emph{covered} by $C$ if there is a candidate $c \in C$ such that $c \succeq_v a_v$, and \emph{uncovered} by $C$ otherwise. 
\end{definition}

Equipped with these notations, we are ready to upper-bound the number of voters preferring an outside candidate to a committee.


\begin{lemma}\label{lemma:02} Given a LEO $(\x,\y,\prices)$  with  $B=1$, income distribution $\threshold$, 
and $C \subseteq A$, for any $a \in A \setminus C$, the number of voters $v \in V$ that strictly prefer $a$ to all candidates in $C$, i.e., $a \succ_v C$, is at most
    $$
    \frac{\beta n}{1-\varepsilon} + |\{v\in \V: v \text{ is uncovered by  } C\}|.
    $$    
\end{lemma}

\begin{proof}
Under a LEO $(\x,\y,\prices)$, the expected revenue obtained by the producer is equal to the expected total price of all voters' random demand. This is because by Claim \ref{cl:lottery}, we have
\[\sum_{a\in A} \left(\sum_{v\in V}p_{v,a}\right) y_a = \sum_{a\in A} \sum_{v\in V}p_{v,a} x_{v,a} =
 \sum_{v\in V} \sum_{a\in A} p_{v,a} x_{v,a} .\]

Note that for each voter $v$, the expression $\sum_{a \in A} p_{v,a} x_{v,a}$ represents the expected cost of their random demand under the LEO lottery. This expected cost must not exceed their expected income $\mathbb{E}_{X \sim \threshold}[X]$, which is at most $\beta$.
Thus
$$
\sum_{a\in A} \left(\sum_{v\in V}p_{v,a}\right) y_a \le n\beta.
$$

First, we show that the total price assigned to each candidate \(a \in A\), i.e., \(\sum_{v \in V} p_{v,a}\), is at most \(n\beta\). Suppose, for the sake of contradiction, that there exists some \(a \in A\) such that \(\sum_{v \in V} p_{v,a} > n\beta\). Then the producer could choose the deterministic allocation with \(y_a = 1\) and \(y_{a'} = 0\) for all \(a' \ne a\). The expected revenue of this allocation would exceed \(n\beta\), which contradicts the fact that \(\y\) is the revenue-maximizing lottery. Therefore, for all \(a \in A\), we must have \(\sum_{v \in V} p_{v,a} \le n\beta\).

Notice that if \(v\) is covered by \(C\) and \(a \succ_v C\), then $a \succ_v a_v$. By Claim \ref{cl:high-price}, $p_{v,a} > 1- \varepsilon$. 
By the argument above we have that the total price for each candidate $a \in A$,  is at most \(n\beta\). 
Therefore, among the covered voters, at most \(\frac{\beta n}{1-\varepsilon}\) prefer \(a\) over \(C\). In the worst case, all the uncovered voters prefer \(a\) over \(C\). This implies Lemma~\ref{lemma:02}.
\end{proof}

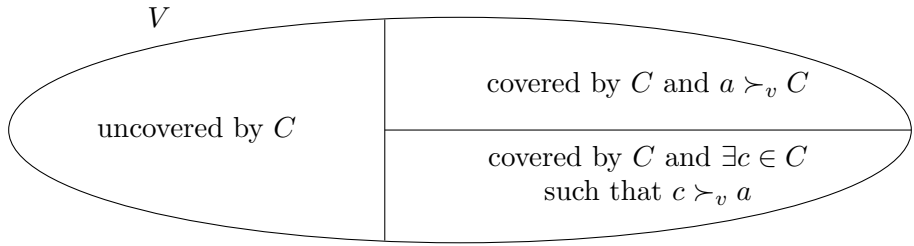
\begin{figure}[h]
    \centering
\begin{tikzpicture}
        \draw (0:0) ellipse (6 and 1.5);
        \node at (-4,1.5) {$V$};
        \draw (-1,1.47) -- (-1,-1.47);
        \draw (6,0) -- (-1,0);
        \node at (-3.5,0) {uncovered by $C$};
        \node[align=center] at (2.5,0.6) {covered by $C$ and $a \succ_v C$};
        \node[align=center] at (2.5,-0.6) {covered by $C$ and $\exists c \in C$ \\ such that $c \succ_v a$};
\end{tikzpicture}
    \caption{Voters $v \in V$ in 3 categories based on a fixed $a \in A \setminus C$: (1) uncovered by $C$, (2) covered by $C$ and $a \succ_v C$, or (3) covered by $C$ and $\exists c \in C$ such that $c \succ_v a$. Note that there are at most $\frac{\beta n}{1-\varepsilon}$ voters in category (2) by Lemma \ref{lemma:02}.}
    \label{fig:enter-label}
\end{figure}



Now we are ready to prove Theorem \ref{thm:undominate}.

\begin{proof}[Proof of Theorem \ref{thm:undominate}.]
Let $(\x,\y,\prices)$ be a LEO with $B=1$ and $\randbud{}=\threshold$. From Claim \ref{cl:lottery},
we regard $\y$ as a lottery for selecting a single candidate. We sample $k$ times independently according to $y$. Let the chosen candidate in the $i$-th lottery be $a_i$. Let $\Tilde{C} = \{a_i \mid i \in [k]\}$ be the random candidate set.

We now bound the probability that $\Tilde{C}$ does not cover $v$, for every voter $v\in \V$.
Recall that $a_v = \max_{\succ_v} \{a \in A \cup \{\emptyset\} \mid p_{v,a} \le 1 - \varepsilon\}$. 

 Observe that when a voter's realized income is at least \(1 - \varepsilon\), voter \(v\) can afford alternative \(a_v\) or some more preferred option. Since this occurs with probability at least \(\beta\), it follows that the probability that a randomly selected alternative from the random demand of $v$, which is  the lottery  $\x_v$, is strictly worse than \(a_v\) is at most \(1 - \beta\).

By Claim~\ref{cl:lottery}, the common allocation $\y$ coincides with each voter's individual random demand $\x_v$. Hence, the probability that a random candidate drawn from $\y$ is worse than $a_v$ according to $\succ_v$ is at most $1-\beta$.
Since $\Tilde{C}$ picks $k$ candidates independently, the probability that $\Tilde{C}$ does not cover $v$ is at most $(1-\beta)^k$.
 
Thus, in expectation over $\Tilde{C}$, there are at most $(1-\beta)^k n$ uncovered voters. Hence, there must be a subset of candidates $C \subseteq A$ of size at most $k$, such that at most $(1-\beta)^k n$ voters are uncovered. By Lemma~\ref{lemma:02}, for every $a\in A \setminus C$, the number of voters strictly preferring $a$ to $C$ is at most $(\frac{\beta}{1-\varepsilon}+ (1-\beta)^k)n$. Because $\frac{\beta}{1-\varepsilon} \to \beta$ as $\varepsilon \to 0$, there is a small $\varepsilon > 0$, so that $\lfloor (\frac{\beta}{1-\varepsilon}+ (1-\beta)^k)\cdot n \rfloor= \lfloor (\beta+ (1-\beta)^k)\cdot n \rfloor$. With this, the proof is completed.
\end{proof}

\section{Generalizing Condorcet Winning Sets}\label{sec:generalized_cond}

This section studies a generalization of the Condorcet winning set. Specifically, we define a new comparison between a set of candidates and a single candidate outside the set, based not on the best candidate in the set but on the $t$-th best. We formalize this idea in the following definition.

\begin{definition}
Given $1\le t \le n$, a voter $v \in V$, a subset of candidates $C \subseteq A$ with $|C| \ge t$, and a candidate $a \in A \setminus C$,  
we say that voter $v$ \emph{$t$-prefers} $C$ to $a$, and write $C \succ_v^t a$,  
if there are at least $t$ candidates in $C$ that $v$ strictly prefers to $a$; that is,
\[
C \succ_v^t a \quad \text{if and only if} \quad \left| \{ c \in C : c \succ_v a \} \right| \ge t.
\]
Otherwise, we say that $v$ $t$-prefers $a$ to $C$, and write $a \succ_v^t C$.
\end{definition}
 
\begin{definition}
        A subset $C \subseteq A$ with $|C|\ge t$ is a $(t,\alpha)$-undominated set if for every candidate $a \in A \setminus C$, 
$$\lvert\{v \in V :  a \succ_v^t C \}\rvert \le  \lfloor \alpha n \rfloor .$$ 
\end{definition}

A special case is when $t = 1$: in this case, $C \succ^1_v a$ means that there is at least one candidate in $C$ that is preferred to $a$, while $a \succ^1_v C$ means that $a$ is preferred to all candidates in $C$.  
Thus, when $t = 1$, the definition of a $(t,\alpha)$-undominated set coincides with the notion of $\alpha$-undominance introduced in Section~\ref{sec:condorcet}.

Our main result of this section is the following.

\thmgeneralupperbound*

Moreover, we provide a lower bound for the size of $(t,\alpha)$-undominated sets, derived by adapting the example provided in \cite{sixCandidates}. The construction is presented in Appendix \ref{app:pf-thm:k-alpha-lower}.

\thmgenerallowerbound*

Given Theorem~\ref{thm:k-alpha-lower}, Theorem~\ref{thm:generalized-cond} is asymptotically tight for every $\alpha$ as $t$ becomes large; that is, the size of a $(t,\alpha)$-undominated set is on the order of $t/\alpha$. For a small value of $t$, the expression for $\delta(t)$, which is derived from the Chernoff bound, is somewhat complex. The exact formulation is provided in the proof at the end of this section. The function $\delta(t)$ steadily decreases as $t$ grows, a trend which can be seen in the following table.

\begin{table}[!htb]
\begin{center}
    \begin{tabular}{|c|c|c|c|c|c|c|c|}
        \hline
        $t$  & 2 & 3 & 4 & 5 & 6 & 7 & 8 \\
        \hline
        $\delta(t)$  & 4.75 & 4.11 & 3.69 & 3.41 & 3.19 & 3.02 & 2.89 \\
        \hline
    \end{tabular}
\end{center}
\caption{Value of $\delta(t)$ in Theorem \ref{thm:generalized-cond}.}
\label{table:transposed-generalized-size}
\end{table}


Before presenting the proof, we first highlight the challenges in generalizing the result from the previous section.
A natural but flawed approach is to iteratively apply the $t = 1$ result: construct an $\alpha$-undominated subcommittee for $t = 1$, remove its members, and repeat this process $t$ times. The union of these $t$ subcommittees has size $O(t)$, and it may seem that for any candidate outside this union, there are $t$ candidates—one from each subcommittee—who are preferred to it. However, this reasoning breaks down. While each subcommittee may contain a candidate preferred to the outside candidate by many voters, the sets of such voters can vary across subcommittees. Their intersection may be small or even empty. Consequently, we cannot guarantee that a sufficiently large number of voters prefer at least $t$ candidates in the union over the outside candidate.\footnote{Moreover, even if this argument were valid, it would not yield a tight bound. As we show, the bound in our theorem becomes asymptotically tight as $t$ grows large: the required committee size approaches the optimal scaling of $t/\alpha$.}

To prove Theorem~\ref{thm:generalized-cond}, we adopt a strategy similar to that used in Section~\ref{sec:condorcet}. We begin with a LEO-style fractional solution and iteratively construct an integral solution from it. However, a direct application of the LEO approach from the previous section does not suffice in our setting. We will need to make three 
key modifications.

First,  in order to ensure a bound on the number of voters who have fewer than \( t \) good candidates in the selected set, we must modify the LEO so that voters “consume” multiple candidates in the equilibrium.  In standard LEO, each voter consumes at most one unit of candidates in total. To address this, we introduce a scaled version of LEO, which we call \(s\)-SLEO. This variant replaces the condition \(\x_v \le \y\) with \(s \x_v \le \y\), for some scaling factor \(s = O(t)\). This modification ensures that each voter effectively allocates their spending across their top \(t\) candidates.

Second, the fractional solution obtained is no longer a lottery over individual candidates, since we are not selecting a single candidate but a set of candidates. Therefore, we need to convert the fractional solution into a lottery over sets of candidates. We achieve this using dependent rounding with negative correlation. 

Third, unlike in the previous section where sampling was independent, we now use an adaptive, iterative procedure to optimize the constant factor in the approximation.



\subsection*{Scaled LEO and Dependent Rounding}
We modify LEO to a version called scaled LEO (SLEO) as follows. First, we  modify the optimal strategy for the producer: 
$$
\mathcal{P}^S(B,\prices):= \arg\max_{\z \in \mathbb{R}_+^{|A|}} \sum_{a\in A}(\sum_{v\in \V} p_{v,a}) z_a  \text{ subject to } \mathbf{1}^T \z = B \text{ and } \z \leq \mathbf{1}.
$$


Note that in this setting, unlike LEO, we explicitly enforce that the producer allocates \emph{at most one unit} of each candidate. We introduce the following $s$-SLEO notation. 




\begin{definition} \label{def:sSLEO}
Given a random income $\randbud{}$ supported on $[0,1]$, a total budget $B>0$, and $s > 0$, an $\scl$-SLEO $(\x, \y, \prices)$ consists of personalized prices $\prices_{v} \in [0,1]^{|A \cup \{\emptyset\}|}$ with  $p_{v,\emptyset} = 0$ and individual consumptions $\x_v \in [0,1]^{|A \cup \{\emptyset\}|}$ for each voter $v \in V$, and a common allocation $\y \in [0,1]^{|A|}$ such that
\begin{enumerate}[label=(\arabic*)]
    \item \label{ssleo-1} $x_{v,a}=  \Pr[\mathcal{D}_{v}(\prices_v, \randbud{}) =a]$ for each $a\in A \cup \{\emptyset\}$, $v\in V$;
    \item \label{ssleo-2} For each $a\in A$, $\scl\cdot x_{v,a} \leq  y_a $, and if the strict inequality holds then  $p_{v,a} = 0$;
    \item \label{ssleo-3} $\y\in   \mathcal{P}^S(B,\prices)$.
\end{enumerate}
\end{definition}

We put the proof of the existence of $\scl$-SLEO in Appendix \ref{app:pf-thm:sSLEO}, which is a slight adjustment of the proof of Theorem \ref{thm:existence} for the existence of LEO.

\begin{restatable}{theorem}{thmssleo} \label{thm:sSLEO}
Let $B \le |A|$, $\scl > 0$, and $\randbud{}$ be a random income with support in the unit interval and its cumulative distribution function is continuous, then there exists an $s$-SLEO.
\end{restatable}

It is important to note that an $\scl$-SLEO yields a fractional solution for each candidate $a\in A$ separately. We will convert it into a lottery over committees of size either $\lceil B \rceil$ or $\lfloor B\rfloor$ using the following dependent rounding process described in \cite{dependent06}.


\begin{proposition}\label{proposition:dep} \cite{dependent06}
Given $\y \in [0,1]^m$ s.t. $\mathbf{1}^T \y = B$, there exists a distribution $\widetilde{Y}$ for a binary vector $\tilde{\y} \in \{0,1\}^m$ satisfying the following properties: 
\begin{enumerate}[label=(\arabic*)]
    \item \label{dep-1} Preservation of marginals: $\underset{\tilde{\y} \sim \widetilde{Y}}{\mathbb{E}}[\Tilde{y}_k] = y_k$ for all $k\in [m]$.
    \item \label{dep-2} Preservation of weights: $\mathbf{1}^T \Tilde{\y} \in \{ \lceil B \rceil, \lfloor B\rfloor\}$.
    \item \label{dep-3} Negative correlation between entries: for any $W \subseteq [m]$ , 
$\underset{\tilde{\y} \sim \widetilde{Y}}{\Pr}\big[ \bigwedge_{k\in W}\Tilde{y}_k = 0 \big] \leq \underset{k\in W}{\prod}\big(1-y_k\big)$. 
\end{enumerate}
\end{proposition}


The foundation of our construction is the combination of the \(s\)-SLEO and a dependent rounding procedure applied to the resulting fractional allocation.

\begin{definition}[Randomized \((B,s,\varepsilon)\)-Lindahl Committee] \label{def:rand-com}
Let \( s, B, \varepsilon >0 \). A \emph{randomized \((B,s,\varepsilon)\)-Lindahl committee} is a random subset \( C \subseteq A \) of size at most $\lceil B \rceil$  constructed as follows:
\begin{enumerate}
    \item Compute an \( s \)-SLEO with budget \( B \), where each voter's income distribution is the uniform distribution over the interval \([1 - \varepsilon, 1]\), denoted as $\mathbf{U}[1 - \varepsilon, 1]$.  Let \( \y \in [0,1]^{|A|} \) be the resulting common fractional allocation over the set of alternatives \( A \). 
    
    \item Apply a dependent rounding procedure to \( \y \) to obtain a committee $C$ of size at most $\lceil B \rceil$ by Proposition \ref{proposition:dep}. We denote this distribution on $C$ by \( \widetilde{C} \).
\end{enumerate}
\end{definition}

A key property of the randomized Lindahl committee is that, when we choose the budget \( B \) and the scaling factor \(\scl\) to be of order \(O(t)\), then with high probability, each voter will have a set of \(t\) \emph{good} candidates in the selected committee.  What does it mean for a set of candidates to be \emph{good} for a voter \(v\) under a realized committee $C$? Informally, it means that if there exists a candidate \(a \notin {C}\) whom \(v\) prefers over these \(t\) candidates in \({C}\), then the individual price of \(a\) in the SLEO satisfies \(p_{v,a} \geq 1 - \varepsilon\).  Why is this useful? We will show that if a candidate \(a\) is excluded from the randomized Lindahl committee with positive probability, then its total price across all voters is bounded by the average total price per candidate, which is at most \(\scl \cdot \frac{n}{B}\). This allows us to derive a bound on the number of voters who \(t\)-prefer a candidate outside the committee.

We start with the first claim bounding the probability that a voter has less than $t$ good candidates in the random Lindahl committee.

\begin{claim}\label{claim:rounding-probability}

Let \( t \geq 2 \) be an integer and \( \gamma \geq 1 \). Let \((\mathbf{x}, \mathbf{y}, \mathbf{p})\) be a \(\gamma t\)-SLEO with random income distribution \(\randbud{} = \mathbf{U}[1 - \varepsilon, 1]\), and let the  budget satisfy \( B \geq \gamma t \). Let \(\widetilde{C}\) be the distribution of the randomized \((B,\gamma t,\varepsilon)\)-Lindahl committee derived from Definition \ref{def:rand-com}. Then, for every voter \( v \in V \), if \( x_{v,\emptyset} = 0 \), the following holds:\footnote{We note that $\randbud{} = \mathbf{U}[1-\varepsilon,1]$ is used here to ensure the consistency with the randomized \((B,\gamma t,\varepsilon)\)-Lindahl committee derived from Definition \ref{def:rand-com}. Claim \ref{claim:rounding-probability} would still hold if the CDF of $\randbud{}$ is continuous and the rounding scheme on $\y$ satisfies the properties in Proposition \ref{proposition:dep}.}
$$
\underset{C\sim \widetilde{C}}{\Pr}\bigg[\bigg| C\cap \{a\in A:x_{v,a}>0\}  \bigg|< t\bigg]\leq \omega(\gamma, t), \text{ where }
\omega(\gamma,t) := \left(\frac{e^{-\frac{\gamma t-t+1}{\gamma t}}}{(\frac{t - 1}{\gamma t})^{\frac{t-1}{\gamma t}}}\right)^{\gamma t}.
$$


\end{claim}
\begin{proof}

For ease of notation, let $A^+:=\{a \in A: x_{v,a} > 0\}$. Let $\widetilde{Y}$ be the distribution of the 0-1 vectors corresponding to the distribution of the randomized committee $\widetilde{C}$, then we have 
$$\underset{C\sim \widetilde{C}}{\Pr}\bigg[| C \cap A^+ |< t  \bigg] 
= \underset{\Tilde{\y}\sim \widetilde{Y}}{\Pr}\bigg[\underset{a \in A^+}{\sum} \tilde{y}_a < t
\bigg].
$$

Because each voter consumes a lottery over $A\cup \{\emptyset\}$ according to condition \ref{ssleo-1} in Definition \ref{def:sSLEO}, if $x_{v,\emptyset} = 0$ then $\underset{a\in A}{\sum}x_{v,a} = 1$. According to condition \ref{ssleo-2} in Definition \ref{def:sSLEO}, $y_a\geq \gamma t\cdot x_{v,a}$ for each $a\in A$. Then we have 
$$\underset{a \in A^+}{\sum} y_a \geq \gamma t \cdot \underset{a\in A}{\sum}x_{v,a}  =\gamma t.$$
Hence, by condition \ref{dep-1} in Proposition \ref{proposition:dep},
$$
   \underset{\tilde{\y}\sim \widetilde{Y}}{\mathbb{E}} \left[\underset{a \in A^+}{\sum} \tilde{y}_a \right] = \underset{a \in A^+}{\sum} y_a \ge \gamma t.
$$
We also have for any $\mu>0$:
$$
\underset{\tilde{\y}\sim \widetilde{Y}}{\Pr}\bigg[\underset{a \in A^+}{\sum} \tilde{y}_a < t
\bigg]= \underset{\tilde{\y}\sim \widetilde{Y}}{\Pr}\bigg[\underset{a \in A^+}{\sum} \tilde{y}_a \le t-1 \bigg]=  \underset{\tilde{\y} \sim \widetilde{Y}}{\Pr}\bigg[\underset{a \in A^+}{\sum} \tilde{y}_a \leq \left(1-\frac{\mu-(t-1)}{\mu}\right) \mu\bigg].$$
The first equality holds because $t$ is integral and $\tilde{y}_a\in \{0,1\}$. 

Because \(\widetilde{Y}\) is negatively correlated over all coordinates according to condition \ref{dep-3} in Proposition \ref{proposition:dep}, we can apply the Chernoff inequality to any subset of coordinates. In particular, we apply it to the subset \(A^+\). We use the Chernoff inequality  for $Z:= \underset{a \in A^+}{\sum} \tilde{y}_a$ in the following form\footnote{See for example \cite{mitzenmacher1995probability}.}: 
$$
\text{for any } \delta \in (0,1) \text{ and } \mu \le \underset{\tilde{y} \sim \widetilde{Y}}{\mathbb{E}}[Z] : \;\;\Pr [Z\leq (1-\delta )\mu ] \leq \left({\frac {e^{-\delta }}{(1-\delta )^{1-\delta }}}\right)^\mu.
$$

Applying Chernoff inequality for  $Z:= \underset{a \in A^+}{\sum} \tilde{y}_a$ with $\mu=\gamma t$, $\delta = \frac{\mu-(t-1)}{\mu}=\frac{\gamma t-(t-1)}{\gamma t}$, we obtain
\[
\underset{\tilde{\y} \sim \widetilde{Y}}{\Pr}\bigg[\underset{a \in A^+}{\sum} \tilde{y}_a \leq \left(1-\frac{\mu-(t-1)}{\mu}\right)\mu\bigg]\leq \omega(\gamma,t).\]
\end{proof}


Similar to Section \ref{sec:condorcet}, we introduce the notion of \emph{boundary candidates}.

\begin{definition} \label{def:boundary-sleo}
    Given an $s$-SLEO $(\x,\y,\prices)$, for each voter $v \in \V$, the \emph{boundary candidate} $a_v$ is the most preferred candidate with price at most $1 - \varepsilon$, that is, $a_v := \max_{\succ_v} \{a \in A \cup \{\emptyset\} \mid p_{a,v} \le 1 - \varepsilon\}$.
\end{definition}

Next, we introduce the definition of $t$-covering.
\begin{definition}\label{def:t-cover}
Given an $\scl$-SLEO $(\x,\y,\prices)$ with the income $\randbud{} \sim \mathbf{U}[1-\varepsilon,1]$, let $C\subseteq \A$ be a set of candidates. We say that a voter $v \in V$ is $t$-\emph{covered} by $C$ if $C\succ_v^t a_v$.
\end{definition}

Using Claim \ref{claim:rounding-probability}, we now show the existence of a committee which $t$-covers a significant number of voters and contains all saturating candidates. 

\begin{claim}\label{claim:deviating_people}
Let \(\gamma \ge 1\) and \(t \ge 2\) be an integer.
Let $(\x, \y, \prices)$ be a $\gamma t$-SLEO with income distribution \(\randbud{} = \mathbf{U}[1 - \varepsilon, 1]\) and  \(B \geq \gamma t\). Then there exists a committee \(C\) of size at most \(\lceil B \rceil\) such that:
\begin{enumerate}
    \item \(C\) fails to $t$-cover at most \(\omega(\gamma, t) \cdot |V|\) voters,
      \item \(C\) contains every alternative \(a\) with \(y_a = 1\).  
\end{enumerate}
\end{claim}

\begin{proof}

Let $C \sim \widetilde{C}$ be a realization of the randomized \((B,\gamma t,\varepsilon)\)-Lindahl committee corresponding to the given $\gamma t$-SLEO. Let $v$ be an arbitrary voter. If $a_v = \emptyset$, then $C$ $t$-covers $v$ with probability 1. This holds because the support of $C$ has size either $\lceil B \rceil$ or $\lfloor B \rfloor$, which is at least $t$. Hence, every realization $C$ is $t$-preferred to $a_v = \emptyset$.

So we suppose $a_v\neq \emptyset$. Notice that whenever a realization $C' \sim \widetilde{C}$ fails to $t$-cover $v$, it must follow that $\bigg| C' \cap \{a\in A:x_{v,a}>0\}  \bigg|< t$.
Then by Claim \ref{claim:rounding-probability}, $C'$ fails to cover $v$ with probability at most $\omega(\gamma,t)$. Then by linearity of expectation, the expected number of agents $C'$ fails to cover is at most $\omega(\gamma,t)\cdot n$. Consequently, there must exist a $C$ in the support of $\widetilde{C}$ which covers no more than $\omega(\gamma,t)\cdot n$ voters. Furthermore, for each $a$ with $y_a = 1$, it must be the case that $a\in C$, because the dependent rounding preserves marginals by condition \ref{dep-1}. This completes the proof. 
\end{proof}

\begin{claim}\label{claim:t-cover}
Given $\gamma > 1$ and $t \in \mathbb{N}$, let $(\x,\y,\prices)$ be a $\gamma t$-SLEO for a set of voters $V$ with income distribution $\randbud{} = \mathbf{U}[1-\varepsilon,1]$ and $B \geq \gamma t$. Let $C$ be an arbitrary committee that includes all candidates $a$ s.t. $y_a = 1$. Let $\mathbb{E}[\randbud{}]:=\mathbb{E}_{b \sim \randbud{}}[b]$ be the expected income. Then for  all $c \in A \setminus C$,  

\begin{enumerate}[label=(\roman*)]
    \item \label{sleo-cover-1} $\underset{v\in V}{\sum}p_{v,c}\leq \gamma t \cdot |V|\cdot \frac{1}{B}\cdot \mathbb{E}[\randbud{}]$. 
    \item \label{sleo-cover-2} There are at most $\frac{\gamma t}{B(1-\varepsilon)}\cdot |V|$ voters who are $t$-covered by $C$ and $t$-prefer $c$ to $C$.
\end{enumerate}
\end{claim}

\begin{proof}
To show \ref{sleo-cover-1}, observe that every voter consumes the random demand under the random income $\randbud{}$. Thus, the expected spending is at most $\mathbb{E}[\randbud{}]$, that is, $\x_v^T \prices_v \le \mathbb{E}[\randbud{}]$.
Condition \ref{ssleo-2} in Definition \ref{def:sSLEO} implies that $x_{v,a} p_{v,a}=  \frac{y_{a} p_{v,a}}{\gamma t}$. Therefore, $\y^T \prices_v = \gamma t \cdot \x^T_v \prices_v \le \gamma t \cdot \mathbb{E}[\randbud{}]$. This implies that the total expected revenue
\[R=\sum_{a \in A} y_a \left( \underset{v\in V}{\sum}p_{v,a} \right) = \sum_{v \in V} \y^T \prices_v \le \gamma t \cdot |V|\cdot \mathbb{E}[\randbud{}].\]

Because \(C\) includes all candidates \(a\) with \(y_a = 1\), for \(c \in A \setminus C\), \(y_c < 1\). Notice that \( y_c \) has not reached its upper limit of 1, which implies that exchanging an \(\varepsilon\)-mass of some \( a' \) with \( y_{a'} > 0 \) for an additional \(\varepsilon\)-mass of \( c \) would not increase the total revenue. It follows that for any $a'$ with $y_{a'}>0$, $\underset{v\in V}{\sum}p_{v,a'}\geq \underset{v\in V}{\sum}p_{v,c}.$ Therefore, if $\underset{v\in V}{\sum}p_{v,c}> \gamma t |V| \cdot \frac{1}{B}\cdot \mathbb{E}[\randbud{}]$, we will have 
$$ R\geq B \cdot \underset{v\in V}{\sum}p_{v,c}>\gamma t \cdot |V| \cdot \mathbb{E}[\randbud{}].$$
This is a contradiction, and therefore we conclude that the aggregate price of $c$ is at most  $\gamma t \cdot |V| \cdot \mathbb{E}[\randbud{}].$


\ref{sleo-cover-2} is a simple consequence of \ref{sleo-cover-1}. By Definition \ref{def:t-cover}, if $C$ $t$-covers a voter $v$ and $c\succ^t_v C$, then $p_{v,c}\geq 1-\varepsilon$. Therefore, among the $t$-covered voters, at most $\frac{\gamma t}{B(1-\varepsilon)}\cdot |V|\cdot \mathbb{E}[\randbud{}]$ of them $t$-prefer $c$ over $C$.
As we choose $\randbud{} = \mathbf{U}[1-\varepsilon, 1]$, $\mathbb{E}[\randbud{}]<1$, and this completes the proof.  
\end{proof}

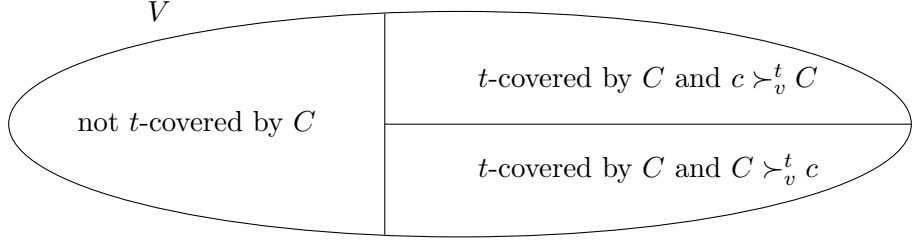
\begin{figure}
    \centering
\begin{tikzpicture}
        \draw (0:0) ellipse (6 and 1.5);
        \node at (-4,1.5) {$V$};
        \draw (-1,1.47) -- (-1,-1.47);
        \draw (6,0) -- (-1,0);
        \node[align=center] at (-3.5,0) {not $t$-covered by $C$};
        \node[align=center] at (2.5,0.6) {$t$-covered by $C$ and $c \succ^t_v C$};
        \node[align=center] at (2.5,-0.6) {$t$-covered by $C$ and $C \succ^t_v c$};
\end{tikzpicture}
    \caption{Voters $v \in V$ in 3 categories based on a fixed $c \in A \setminus C$: (1) not $t$-covered by $C$, (2) $t$-covered by $C$ and $c \succ^t_v C$, or (3) $t$-covered by $C$ and $C \succ^t_v c$. Note that there are at most $\omega(\gamma,t) \cdot |V|$ voters in category (1) by Claim \ref{claim:deviating_people} and at most $\frac{\gamma t}{B(1-\epsilon)} \cdot |V|$ voters in category (2) by Claim \ref{claim:t-cover}.}
    \label{fig:03}
\end{figure}

Claims~\ref{claim:deviating_people} and~\ref{claim:t-cover} are illustrated in Figure~\ref{fig:03}. Together, these results establish the existence of a committee \(C\) of size at most \(\lceil B \rceil\) that limits the number of voters who prefer $c \in A \setminus C$ to $C$. Specifically, Claim~\ref{claim:deviating_people} provides an upper bound of \(\omega(\gamma, t)\cdot |V|\) on the number of voters who are not \(t\)-covered by \(C\). Claim~\ref{claim:t-cover}, in turn, shows that for any candidate $c \in A \setminus C$, among the voters who are \(t\)-covered by \(C\), at most \(\frac{\gamma t}{B(1 - \varepsilon)} \cdot |V|\) of them \(t\)-prefer \(c\) over \(C\).



As in the result from Section~\ref{sec:condorcet}, combining these two bounds yields a bound on the total number of voters who might \(t\)-prefer \(c\) to \(C\). To ensure that this number is less than \(\alpha \cdot |V|\), it is necessary to choose \(\gamma\) and \(B\) as functions of \(\alpha\) and \(t\). Combining Claims \ref{claim:deviating_people} and~\ref{claim:t-cover} implies the following.


\begin{claim}\label{thm:repeated-sampling}
For any $\alpha\in (0,1)$, integer $t \ge 2$, and $\varepsilon>0$, if there exists a  $\gamma\geq 1$ and an integer $B\geq \gamma t$ s.t. $$ \alpha \geq \frac{\gamma t}{B} + \omega(\gamma, t),$$
then a $(t,\frac{\alpha}{1-\varepsilon})$-undominated set of size $B$ exists. 
\end{claim}

Claim~\ref{thm:repeated-sampling} allows us to derive bounds for $(t,\alpha)$-undominated sets for certain values of $t$ and $\alpha$.
In general, when $\alpha$ is relatively large, the bound in Claim~\ref{thm:repeated-sampling} can be used to obtain a committee of size at most $\delta(t) \cdot t$, as stated in Theorem~\ref{thm:generalized-cond}.
Furthermore, as we show in the proof of Theorem~\ref{thm:generalized-cond}, the bound approaches the optimal value of $\frac{t}{\alpha}$ as $t \to \infty$.


However, this method has limitations. It does not guarantee the construction of a $(t,\alpha)$-undominated set of the desired size order \(O(\frac{t}{\alpha})\) when $\frac{1}{\alpha}$ is very big compared to $t$. For example, one can show that in order to be a $(t,\alpha)$-undominated set,  $B$ needs to be in the order $\frac{1}{\alpha}\cdot \log(\frac{1}{\alpha})$, which may well exceed \(O(\frac{t}{\alpha})\) when $\log(\frac{1}{\alpha})$ is relatively large to \(t\). 

To obtain a constant bound for all \( t \) and \( \alpha \), we incorporate SLEO into the framework of \cite{ApproStable} to derive a new iterative method for constructing the committee, which yields a tighter bound when \( \alpha \) is small. By combining the bound from Claim~\ref{thm:repeated-sampling} with the new bound described below, we show that the resulting committee satisfies the desired size requirement of \( O(t/\alpha) \).

\subsection*{Iterative Algorithm}

To overcome the issue described above, we give an iterative process that repeatedly applies Claims~\ref{claim:deviating_people} and~\ref{claim:t-cover}, each time restricting attention to the set of voters who are not \(t\)-covered by the current committee \(C\). Note that the number of such voters decreases by a constant factor in each iteration (specifically, by \(\omega(\gamma, t)\)). The process will continue until `` almost" all voters are covered, thus solving the issue of the one-step rounding.

To maintain control over the total committee size, we also reduce the budget \(B\) by a constant factor at each step, ensuring that the sequence \(\{B_i\}_{i=0}^{\infty}\) forms a geometric progression. This prevents the total size of the output committee from growing unbounded across iterations.

The precise iterative algorithm is described in Algorithm~\ref{alg:iterative}.


\begin{algorithm} 
\caption{Iterative Rounding with SLEO} \label{alg:iterative}
\KwData{A set $V$ of $n$ voters, $\alpha\leq 1$, integer $t \ge 2$, $\gamma, \tau \geq 1$ with $\omega(\gamma,t)\tau< 1$,  $\varepsilon \in (0,1)$} 
\KwResult{A $(t,\frac{\alpha}{1-\varepsilon})$-undominated set $C$ }
$V_0\leftarrow V$, $t\leftarrow 0$, $B_0 \gets \gamma\cdot \frac{1}{1-\omega(\gamma, t)\cdot \tau}\cdot \frac{t}{\alpha}$\\
\While{$|B_i|\geq \gamma t$}{
$B_i\leftarrow \frac{B_0}{\tau^i}$\\
Generate a $\gamma t$-SLEO $(\x,\y,\prices)$ for $V_i$ with budget $B_i$ and random income $\randbud{} = \mathbf{U}[1-\varepsilon,1]$\\
Find a set $C_i$ of size at most $\lceil B_i\rceil$ which:\\
\quad (i) contains all candidates $a$ with $y_a = 1$ and \\
\quad (ii) does not cover at most $\omega(\gamma,t)\cdot |V_i|$ voters  \\
$V_{i+1}\leftarrow V_i\setminus\{v \in V_{i}: C_i \text{ $t$-covers }v  \}$\\
$C\leftarrow C\cup C_i$\\
$i\leftarrow i+1$
}
return $C$
\end{algorithm}

The parameters \(\gamma, \tau \geq 1\) are selected later to minimize $\delta$ in Theorem \ref{thm:generalized-cond}.

The algorithms begin by considering the full set of voters and computing a \(\gamma t\)-SLEO with a total budget \(B_0\). It then converts this outcome into a randomized Lindahl committee. By Claim~\ref{claim:deviating_people}, there exists a realization of this randomized committee that fails to cover a certain fraction of the voters. We then restrict our attention to the uncovered voters and apply the same procedure, but with a reduced budget: \(B_1 = B_0 / \tau\). This step may again leave a portion of voters uncovered.

This process is repeated iteratively, each time applying the method to the currently uncovered voters and reducing the budget geometrically. The iteration continues until the current budget \(B_i\) becomes smaller than \(\gamma t\), at which point the process terminates.

To see the correctness of the algorithm, a crucial observation is that because $|V_i|$ decrease \textit{simultaneously} with $B_i$ for each candidate $a$, the number of $t$-covered voters preferring $a$ to $C_i$ also forms a geometric sequence.

\begin{claim}\label{lemma:geometric}
For each step  $i \ge 0$ of Algorithm~\ref{alg:iterative}, let $S(C_i,V_i)$ denote the set of voters among $V_i$ t-covered by $C_i$. Then for any $a\notin C_i$, it holds
$$ |v\in S(C_i,V_i): a\succ_v^t C_i|\leq \frac{\gamma t}{1-\varepsilon}\cdot [\omega(\gamma,t)\tau]^i\cdot \frac{|V|}{B_0}.$$ 
\end{claim}

\begin{proof}
According to Claim \ref{claim:deviating_people}, for each $j$, we have 
$$|V_{i+1}|\leq \omega(\gamma,t)\cdot |V_i|.$$ 
Therefore, at round $i$, we can derive through induction that 
$$|V_i|\leq \omega(\gamma,t)^i \cdot |V|.$$ And by the construction of $C_i$, $a\notin C_i$ must have $y_a>0$. Then the lemma follows from Claim \ref{claim:t-cover}:
$$ |v\in S(C_i,V_i): a\succ_v^t C_i|\leq \frac{\gamma t}{B_i(1-\varepsilon)}\cdot |V_i| = \frac{\gamma t}{1-\varepsilon}\cdot [\omega(\gamma,t)\tau]^i\cdot \frac{|V|}{B_0}.$$
\end{proof}

\begin{restatable}{claim}{clmalgorithmcorrectness}\label{clm:algorithm-correctness}
Given a set $V$ of $n$ voters, $\alpha \leq 1$, integer $t \ge 2$,$\gamma,\tau\geq 1$ with $\omega(\gamma,t)\tau< 1$, and $\varepsilon>0$, the output of Algorithm \ref{alg:iterative} is a committee which is $(t,\frac{\alpha}{1-\varepsilon})$-undominated and of size at most $\frac{\gamma \tau}{\tau-1}\cdot\frac{1}{1-\omega(\gamma,t) \tau}\cdot \frac{t}{\alpha} + \log_{\tau}(\frac{1}{1-\omega(\gamma, t) \tau}\cdot \frac{1}{\alpha})$ . 
\end{restatable}

The proof of correctness of the algorithm, which is mostly algebra of the iterative process, can be found in Appendix \ref{app:pf-clm:algorithm-correctness}. 

Below, we provide the proof of Theorem~\ref{thm:generalized-cond}. The function $\delta(t)$ is derived by optimizing the bounds established in Claim~\ref{clm:algorithm-correctness} and Claim~\ref{thm:repeated-sampling}, and then taking their minimum. The undominance ratio $\alpha$ is obtained by letting \(\varepsilon \to 0\).

\begin{proof}[Proof of Theorem \ref{thm:generalized-cond}] 


We assume that $\delta(t)\cdot \frac{t}{\alpha}<|A|$, otherwise the theorem trivially holds as we take the entire set $A$ as the outcome set.

Now define
\begin{equation} \label{eq:delta}
    \delta(t) := \max_{\alpha\in(0,1]}\min\{\frac{s_1(\alpha,t)}{t/\alpha}, \frac{s_2(\alpha,t)}{t/\alpha}\}
\end{equation}
where
\begin{align*}
    s_1(\alpha,t)&:= \min_{\gamma, \tau\geq 1} \frac{\gamma\cdot \tau}{\tau-1}\cdot \frac{1}{1- \omega(\gamma,t)\tau}\cdot \frac{t}{\alpha} + \log_{\tau}(\frac{1}{1-\omega(\gamma, t) \tau}\cdot \frac{1}{\alpha})  \text{ subject to } \tau \omega(\gamma,t)<1, \\
    s_2(\alpha,t)&:= \min_{\gamma\geq 1, B\in\mathbb{N}} B  \text{ subject to } \alpha \geq\frac{\gamma t}{B} + \omega(\gamma, t).
\end{align*}
The plots of \( s_1 \) and \( s_2 \), along with their comparison, are deferred to the Appendix.
\ref{app:plot}. 

By Claim \ref{clm:algorithm-correctness} and Claim \ref{thm:repeated-sampling}, for  an $\varepsilon>0$, there exists a committee $C_{\varepsilon}$ of size $\delta(t)\cdot \frac{t}{\alpha}$ s.t. for any $c \in A \setminus C_\varepsilon$, the number of voters 
$t$-preferring $c$ to $C_{\varepsilon}$  is at most $\frac{1}{1-\varepsilon}\cdot \alpha n$. Because the number of voters is integral, it is at most  $\lfloor \frac{1}{1-\varepsilon}\cdot \alpha n \rfloor$.

Observe that $\frac{1}{1-\varepsilon}\cdot \alpha n > \alpha n $ and it approaches $\alpha n$  as $\varepsilon$ goes to $0$, hence, there is a small enough $\varepsilon^*$ so that  $\lfloor \frac{1}{1-\varepsilon}\cdot \alpha n \rfloor= \lfloor \alpha n \rfloor$ for any $\varepsilon\le \varepsilon^*$. 

This implies that the set \( C_{\varepsilon^*} \) satisfies our desired property: for any \( c \notin C_{\varepsilon^*} \), the number of voters who \( t \)-prefer \( c \) to \( C_{\varepsilon^*} \) is at most \( \alpha n \).



We leave in Appendix \ref{app:pf-thm:generalized-cond} the parts of the proof showing that $\delta(t)\leq 4.75$ for any $t\geq 2$ and $\delta(t)\rightarrow 1$ when $t\rightarrow \infty$. \end{proof}

\section{Conclusion}
This paper systematically studies the problem of selecting undominated committees in social choice, aiming to overcome the Condorcet paradox by allowing for the selection of multiple candidates. Our framework not only generalizes the majority condition to accommodate any arbitrary fraction of voter support, but more importantly, introduces a richer comparison principle: evaluating an outside candidate against the \(t\)-th most-preferred member of the selected committee, rather than just the top-ranked one. This generalized notion has practical implications, providing a more nuanced basis for excluding candidates through comparisons with median or quantile-level members of the committee.  It also reveals a striking structural property: in the limit, a tight lower bound emerges that reflects the dual role of the threshold parameter, simultaneously capturing the breadth of voter support and the depth of preference comparisons.

Future work includes deriving tighter bounds for finite $t$ and different ways of comparing a set of candidates to a single one. More broadly, our results are based on a novel approach grounded in Lindahl equilibrium with ordinal preferences and its extensions. This technique is general and flexible, making it applicable beyond the specific context of committee selection. For instance, it has been applied to   participatory budgeting (\cite{nguyen-song}), where resources must be allocated fairly across competing projects without relying on cardinal utilities. It may also be applied to public goods provision, where the goal is to ensure fair and efficient outcomes based on individuals’ ranked preferences over outcomes. Furthermore, our framework has potential implications in algorithmic fairness, particularly in settings where decisions must be made under ordinal input data, such as in hiring, admissions, or recommendation systems. By bridging concepts from market equilibrium and ordinal social choice, the approach offers a unifying lens for designing fair mechanisms across a wide range of collective decision-making problems.

\newpage
\appendix
\section{Proof of Theorem \ref{thm:existence}} \label{app:LEO}

We choose to include renowned Kakutani's fixed point theorem and Maximum theorem here for the sake of completeness. 
\begin{theorem}[Kakutani's fixed point theorem, \cite{Kakutani41}]\label{thm:kakutani}
Let $S$ be a non-empty, compact, and convex subset of some Euclidean space $\mathbb{R}^n.$ Let $\psi:S\rightarrow 2^S \setminus \emptyset$ be a point-to-set function on $S$ such that 1) $\psi$ is upper-hemicontinuous and 2) $\psi(s)$ is non-empty and convex for all $s \in S$. Then there exists a \emph{fixed point} $s \in S$ such that $s \in \psi(s)$. 
\end{theorem}

\begin{theorem}[Maximum Theorem, \cite{maximum}]\label{thm:maximum}
Let $\mathbf{X}$ and $\mathbf{\Theta}$ be topological spaces, $f:\mathbf{X}\times \mathbf{\Theta}\rightarrow \mathbb{R}$ be a continuous function on the product $\mathbf{X}\times\mathbf{\Theta}$, and $\mathbf{C}:\mathbf{\Theta}\rightrightarrows \mathbf{X}$ be a compact valued correspondence such that $\mathbf{C}(\theta)\neq \emptyset$ for all $\theta \in\mathbf{\Theta}$. Define $f^*(\theta) = \sup\{f(x,\theta):x\in \mathbf{C}(\theta)\}$ and the correspondence $\mathbf{C}^*: \mathbf{\Theta} \rightrightarrows \mathbf{X}$ by $\mathbf{C}^*(\theta) = \{x \in \mathbf{C}(\theta): f(x,\theta) = f^*(\theta)\}$. If $\mathbf{C}$ is continuous at $\theta$, then $f^*$ is continuous and $\mathbf{C}^*$ is upper hemi-continuous, non-empty, and compact valued. As a consequence, the $\sup$ can be replaced by $\max$. 
\end{theorem}




Equipped with these theorems, we are ready to show Theorem \ref{thm:existence}.

\thmleo*

\begin{proof}

Define $\Gamma:= \{\z \in \mathbb{R}_+^{|A|} \mid \mathbf{1}^T\z = B \} $. We construct the following correspondence:


$$\mathcal{L}: \prod_{v=1}^n[0,1]^{|A \cup \{\emptyset\}|}\times \Gamma \times \prod_{v=1}^n[0,1]^{|A \cup \{\emptyset\}|} \rightrightarrows \prod_{v=1}^n[0,1]^{|A \cup \{\emptyset\}|} \times \Gamma  \times\prod_{v=1}^n[0,1]^{|A \cup \{\emptyset\}|}$$
with
$$\mathcal{L}\left((\{\x_{v}\}_{v\in V},\y,\{\prices_{v}\}_{v\in V})\right) = (\{\x'_{v}\}_{v\in V},Y',\{\prices'_{v}\}_{v\in V})$$ where
\begin{subequations}
\begin{align}
x'_{v,a} &= \Pr[\mathcal{D}_{v}(\prices_{v}, \randbud{}) =a] & \forall v \in V, a \in A \cup \{\emptyset\}, \label{eq:fp-x} \\
Y' &= \arg \max_{\z \in \Gamma} \;\; \left(\sum_{v\in V} \prices_v\right)^T \z, \label{eq:fp-y} \\ 
p'_{v,a} &= 
\begin{cases}
\max\left\{\min\left\{1,p_{v,a}+(x'_{v,a}-y_a)\right\} ,0 \right\} \text{  if } a \in A,\\
0 \text{  if } a = \emptyset,
\end{cases} & \forall v \in V, a \in A \cup \{\emptyset\}. \label{eq:fp-p}
\end{align}
\end{subequations}
Our goal is to show that
\begin{enumerate}[label=(\roman*)]
    \item \label{pf-leo-1} By Theorem \ref{thm:maximum}, the correspondence $\prod_{v=1}^n[0,1]^{|A \cup \{\emptyset\}|} \rightrightarrows \Gamma$ (from $\prices$ to $Y'$) is upper hemi-continuous, non-empty, and compact valued. Furthermore, $Y'$ is convex, bounded, and non-empty.
    \item \label{pf-leo-2} By Theorem \ref{thm:kakutani}, there exists a \emph{fixed point} $(\{\x_{v}\}_{v\in V},\y,\{\prices_{v}\}_{v\in V})$ such that 
    $$(\{\x_{v}\}_{v\in V},\y,\{\prices_{v}\}_{v\in V}) \in \mathcal{L}(\{\x_{v}\}_{v\in V},\y,\{\prices_{v}\}_{v\in V}).$$
    \item \label{pf-leo-3} The fixed point $(\{\x_{v}\}_{v\in V},\y,\{\prices_{v}\}_{v\in V})$ satisfies all the conditions of a LEO in Definition \ref{def:LEO}. 
\end{enumerate}

To show \ref{pf-leo-1}, we employ Theorem \ref{thm:maximum} with $\mathbf{X} = \Gamma$ (the space of the producer's decision) and $\mathbf{\Theta} = \prod_{v=1}^n[0,1]^{|A \cup \{\emptyset\}|}$ (the space of prices $\prices$). Let $f(\z,\prices) = \left(\sum_{v\in V} \prices_v\right)^T \z$ be the producer revenue given the producer decision $\z \in \Gamma$ and the prices $\prices \in \mathbf{\Theta}$. For any prices $\prices \in \mathbf{\Theta}$, any producer decision in $\Gamma$ is feasible, so $\mathbf{C}(\prices) = \Gamma$ for any $\prices \in \mathbf{\Theta}$. $\mathbf{C}$ is a compact valued correspondence since $\Gamma$ is compact. $f^*$ maps $\prices$ to the maximum producer revenue, thus, given prices $\prices$, we have $\mathbf{C}^*(\prices) = \arg \max_{\z \in \Gamma} \;\; \left(\sum_{v\in V} \prices_v\right)^T \z$. That is, $\mathbf{C}^*(\prices)$ is the set of the optimal producer decision. By Theorem \ref{thm:maximum}, $\mathbf{C}$ is continuous at any $\prices \in \mathbf{\Theta}$, so $f^*$ is continuous and $\mathbf{C}^*$ is upper hemi-continuous, non-empty, and compact valued. Furthermore, $\mathbf{C}^*(\prices)$ is convex, non-empty, and bounded for any $\prices \in \mathbf{\Theta}$ since it is the set of the optimal solutions of a linear program with a bounded and non-empty feasible region.

To show \ref{pf-leo-2}, we consider $S = \prod_{v=1}^n[0,1]^{|A \cup \{\emptyset\}|}\times \Gamma \times \prod_{v=1}^n[0,1]^{|A \cup \{\emptyset\}|}$ and let $\psi = \mathcal{L}$. The range of $\psi$ is restricted to having $\x'$ and $\prices'$ as a point and $Y' \subseteq \Gamma$ as a set. Clearly, $S$ is non-empty, compact, and convex. For all $s \in S$, $\psi(s)$ is non-empty and convex since $Y'$ is convex, non-empty, and bounded. The remaining is to show that $\mathcal{L}$ is upper hemi-continuous. To show this, we prove that (a) the mapping from $s \in S$ to $\x'$ is continuous, (b) the mapping from $s \in S$ to $Y'$ is upper hemi-continuous, and (c) the mapping from $s \in S$ to $\prices'$ is continuous.
(b) directly holds from \ref{pf-leo-1} and (c) holds once (a) holds since $p'_{v,a}$ defined in \eqref{eq:fp-p} is a continuous function of $\x'$, $\y$, and $\prices$. We focus on showing (a). Let $F_{\randbud{}}$ be the CDF for the budget given to a voter $v \in V$. We have that
\begin{align*}
    x'_{v,a} &= 
\Pr[\mathcal{D}_{v}(\prices_{v}, \randbud{v}) =a] \\
&=  
\begin{cases}
    \underset{a' \in A \cup \{\emptyset\}: a'\succ_v a}{\min}\bigg( F_{\randbud{}}(p_{v,a'}) - F_{\randbud{}}(p_{v,a})\bigg)^+ & \text{if } a \in A \cup \{\emptyset\} \text{ is not the top-ranked in } \succ_v, \\
    1 - F_{\randbud{}}(p_{v,a}) & \text{if } a \in A \cup \{\emptyset\} \text{ is the top-ranked in } \succ_v.
\end{cases}
\end{align*}
Intuitively, if $a$ is the top-ranked in $\succ_v$, then $v$ demands $a$ as long as $p_{v,a}$ is at most the budget given to $v$. The probability of having such a budget is $1 - F_{\randbud{}}(p_{v,a})$. If $a$ is the second-ranked and $a'$ is the top-ranked in $\succ_v$, then $v$ demands $a$ if the budget is below $p_{v,a'}$ and at least $p_{v,a}$. This probability is captured by $\left(F_{\randbud{}}(p_{v,a'}) - F_{\randbud{}}(p_{v,a})\right)^+$. Note that if $p_{v,a'} < p_{v,a}$, then $v$ never demands $a$ since it is more expensive and less preferred. Applying analogous reasoning results in the closed form above for the probability that $v$ demands $a$ under a random budget $\randbud{}$. More specifically, consider $a' \succ_v a$, the probability that $a$ is demanded by $v$ under $\prices_v$ is zero when there exists $p_{v,a'} \le p_{v,a}$. Otherwise, it is the minimum difference between the CDF at $p_{v,a'}$ and $p_{v,a}$. Since $F_{\randbud{}}$ and the closed form above are continuous, (a) holds. By Theorem \ref{thm:kakutani}, there exists $s \in S$ such that $s \in \mathcal{L}(s)$.

To show \ref{pf-leo-3}, we prove that any fixed point $s=(\{\x_{v}\}_{v\in V},\y,\{\prices_{v}\}_{v\in V})$ such that $s \in \mathcal{L}(s)$ satisfies conditions \ref{leo-1}, \ref{leo-2}, and \ref{leo-3} in Definition \ref{def:LEO}. Conditions \ref{leo-1} and \ref{leo-3} follow directly by \eqref{eq:fp-x} and \eqref{eq:fp-y}, respectively. For condition \ref{leo-2}, we first show that $x_{v,a} \le y_a$ for all $a \in A \cup \{\emptyset\}$ and $v \in V$. Suppose for the sake of contradiction that $x_{v,a} > y_a$ for some $a \in A \cup \{\emptyset\}$ and $v \in V$. Since $x_{v,a} - y_a > 0$ and $s$ is a fixed point, from \eqref{eq:fp-p}, we must have $p_{v,a} = 1$. Consequently, $v$ cannot afford $a$ with a strictly positive probability unless $\Pr[\randbud{}=1]$ is strictly positive, violating the fact that $\Pr[\randbud{}=1]=0$. Therefore, $x_{v,a} \le y_a$. Now suppose $x_{v,a} < y_a$, then from \eqref{eq:fp-p}, the fixed point $s$ forces $p_{v,a} = 0$. Hence, $s=(\{\x_{v}\}_{v\in V},\y,\{\prices_{v}\}_{v\in V})$ satisfies all the conditions of Definition \ref{def:LEO}, so $s$ is a LEO. 
\end{proof}

\section{Omitted Proofs in Section \ref{sec:generalized_cond}} \label{app:gen_cond}

\subsection{Proof of Theorem \ref{thm:k-alpha-lower}} \label{app:pf-thm:k-alpha-lower}

\thmgenerallowerbound*

\begin{proof}[]

Let $\ell$ be a large integer. We construct an election with $(k+1) \ell$ voters and $(k+1) \ell$ candidates. For ease of presentation, we define $[\ell] = \{0,1,...,\ell-1\}$ as the residue classes modulo $\ell$ throughout the proof. We will work with modulo arithmetic throughout this proof. 
    
Each voter and each candidate are associated with a unique tuple in $V = [k+1] \times [\ell]$ and $A = [k+1] \times [\ell]$, respectively. The voter $(p,q)$ ranks the candidates in reverse lexicographical order after shifting the first entry down by $p$ and the second entry down by $q$, modulo $k+1$ and $\ell$, respectively. That is, $(x,y) \succ_{(p,q)} (x',y')$ if
    \begin{itemize}
        \item $x-p < x'-p$, or
        \item $x=x'$ and $y-q < y'-q$.
    \end{itemize}
We show an example below.

\begin{table}[h]
    \centering
    \begin{tabular}{l*{20}{C}}
        $v_0$ & $v_1$ & $v_2$ & $v_3$ & $v_4$ & $v_5$ & $v_6$ & $v_7$ & $v_8$ & $v_9$ & $v_{10}$ & $v_{11}$ & $v_{12}$ & $v_{13}$ & $v_{14}$ & $v_{15}$ & $v_{16}$ & $v_{17}$ & $v_{18}$ & $v_{19}$\\
        \hline
        \textcolor{red}{0} & \textcolor{red}{1} & \textcolor{red}{2} & \textcolor{red}{3} & \textcolor{red}{4} & \textcolor{blue}{5} & \textcolor{blue}{6} & \textcolor{blue}{7} & \textcolor{blue}{8} & \textcolor{blue}{9} & \textcolor{ForestGreen}{10} & \textcolor{ForestGreen}{11} & \textcolor{ForestGreen}{12} & \textcolor{ForestGreen}{13} & \textcolor{ForestGreen}{14} & \textcolor{BurntOrange}{15} & \textcolor{BurntOrange}{16} & \textcolor{BurntOrange}{17} & \textcolor{BurntOrange}{18} & \textcolor{BurntOrange}{19} \\
        \textcolor{red}{1} & \textcolor{red}{2} & \textcolor{red}{3} & \textcolor{red}{4} & \textcolor{red}{0} & \textcolor{blue}{6} & \textcolor{blue}{7} & \textcolor{blue}{8} & \textcolor{blue}{9} & \textcolor{blue}{5} & \textcolor{ForestGreen}{11} & \textcolor{ForestGreen}{12} & \textcolor{ForestGreen}{13} & \textcolor{ForestGreen}{14} & \textcolor{ForestGreen}{10} & \textcolor{BurntOrange}{16} & \textcolor{BurntOrange}{17} & \textcolor{BurntOrange}{18} & \textcolor{BurntOrange}{19} & \textcolor{BurntOrange}{15} \\
        \textcolor{red}{2} & \textcolor{red}{3} & \textcolor{red}{4} & \textcolor{red}{0} & \textcolor{red}{1} & \textcolor{blue}{7} & \textcolor{blue}{8} & \textcolor{blue}{9} & \textcolor{blue}{5} &\textcolor{blue}{6} &  \textcolor{ForestGreen}{12} & \textcolor{ForestGreen}{13} & \textcolor{ForestGreen}{14} & \textcolor{ForestGreen}{10} & \textcolor{ForestGreen}{11} & \textcolor{BurntOrange}{17} & \textcolor{BurntOrange}{18} & \textcolor{BurntOrange}{19} & \textcolor{BurntOrange}{15} & \textcolor{BurntOrange}{16} \\
        \textcolor{red}{3} & \textcolor{red}{4} & \textcolor{red}{0} & \textcolor{red}{1} & \textcolor{red}{2} &  \textcolor{blue}{8} & \textcolor{blue}{9} & \textcolor{blue}{5} & \textcolor{blue}{6} & \textcolor{blue}{7} &  \textcolor{ForestGreen}{13} & \textcolor{ForestGreen}{14} & \textcolor{ForestGreen}{10} & \textcolor{ForestGreen}{11} & \textcolor{ForestGreen}{12} & \textcolor{BurntOrange}{18} & \textcolor{BurntOrange}{19} & \textcolor{BurntOrange}{15} & \textcolor{BurntOrange}{16} & \textcolor{BurntOrange}{17} \\
        \textcolor{red}{4} & \textcolor{red}{0} & \textcolor{red}{1} & \textcolor{red}{2} &  \textcolor{red}{3} & \textcolor{blue}{9} & \textcolor{blue}{5} & \textcolor{blue}{6} & \textcolor{blue}{7} & \textcolor{blue}{8} &  \textcolor{ForestGreen}{14} & \textcolor{ForestGreen}{10} & \textcolor{ForestGreen}{11} & \textcolor{ForestGreen}{12} & \textcolor{ForestGreen}{13} & \textcolor{BurntOrange}{19} & \textcolor{BurntOrange}{15} & \textcolor{BurntOrange}{16} & \textcolor{BurntOrange}{17} & \textcolor{BurntOrange}{18} \\
        
        \textcolor{blue}{5} & \textcolor{blue}{6} & \textcolor{blue}{7} & \textcolor{blue}{8} & \textcolor{blue}{9} & \textcolor{ForestGreen}{10} & \textcolor{ForestGreen}{11} & \textcolor{ForestGreen}{12} & \textcolor{ForestGreen}{13} & \textcolor{ForestGreen}{14} & \textcolor{BurntOrange}{15} & \textcolor{BurntOrange}{16} & \textcolor{BurntOrange}{17} & \textcolor{BurntOrange}{18} & \textcolor{BurntOrange}{19} & \textcolor{red}{0} & \textcolor{red}{1} & \textcolor{red}{2} & \textcolor{red}{3} & \textcolor{red}{4} \\
        \textcolor{blue}{6} & \textcolor{blue}{7} & \textcolor{blue}{8} & \textcolor{blue}{9} & \textcolor{blue}{5} & \textcolor{ForestGreen}{11} & \textcolor{ForestGreen}{12} & \textcolor{ForestGreen}{13} & \textcolor{ForestGreen}{14} & \textcolor{ForestGreen}{10} & \textcolor{BurntOrange}{16} & \textcolor{BurntOrange}{17} & \textcolor{BurntOrange}{18} & \textcolor{BurntOrange}{19} & \textcolor{BurntOrange}{15} & \textcolor{red}{1} & \textcolor{red}{2} & \textcolor{red}{3} & \textcolor{red}{4} & \textcolor{red}{0} \\
        \textcolor{blue}{7} & \textcolor{blue}{8} & \textcolor{blue}{9} & \textcolor{blue}{5} & \textcolor{blue}{6} &  \textcolor{ForestGreen}{12} & \textcolor{ForestGreen}{13} & \textcolor{ForestGreen}{14} & \textcolor{ForestGreen}{10} & \textcolor{ForestGreen}{11} & \textcolor{BurntOrange}{17} & \textcolor{BurntOrange}{18} & \textcolor{BurntOrange}{19} & \textcolor{BurntOrange}{15} & \textcolor{BurntOrange}{16} & \textcolor{red}{2} & \textcolor{red}{3} & \textcolor{red}{4} & \textcolor{red}{0} & \textcolor{red}{1} \\
        \textcolor{blue}{8} & \textcolor{blue}{9} & \textcolor{blue}{5} & \textcolor{blue}{6} & \textcolor{blue}{7} &  \textcolor{ForestGreen}{13} & \textcolor{ForestGreen}{14} & \textcolor{ForestGreen}{10} & \textcolor{ForestGreen}{11} & \textcolor{ForestGreen}{12} & \textcolor{BurntOrange}{18} & \textcolor{BurntOrange}{19} & \textcolor{BurntOrange}{15} & \textcolor{BurntOrange}{16} & \textcolor{BurntOrange}{17} & \textcolor{red}{3} & \textcolor{red}{4} & \textcolor{red}{0} & \textcolor{red}{1} & \textcolor{red}{2} \\
        \textcolor{blue}{9} & \textcolor{blue}{5} & \textcolor{blue}{6} & \textcolor{blue}{7} & \textcolor{blue}{8} &  \textcolor{ForestGreen}{14} & \textcolor{ForestGreen}{10} & \textcolor{ForestGreen}{11} & \textcolor{ForestGreen}{12} & \textcolor{ForestGreen}{13} & \textcolor{BurntOrange}{19} & \textcolor{BurntOrange}{15} & \textcolor{BurntOrange}{16} & \textcolor{BurntOrange}{17} & \textcolor{BurntOrange}{18} & \textcolor{red}{4} & \textcolor{red}{0} & \textcolor{red}{1} & \textcolor{red}{2} &  \textcolor{red}{3} \\

        \textcolor{ForestGreen}{10} & \textcolor{ForestGreen}{11} & \textcolor{ForestGreen}{12} & \textcolor{ForestGreen}{13} & \textcolor{ForestGreen}{14} & \textcolor{BurntOrange}{15} & \textcolor{BurntOrange}{16} & \textcolor{BurntOrange}{17} & \textcolor{BurntOrange}{18} & \textcolor{BurntOrange}{19} & \textcolor{red}{0} & \textcolor{red}{1} & \textcolor{red}{2} & \textcolor{red}{3} & \textcolor{red}{4} & \textcolor{blue}{5} & \textcolor{blue}{6} & \textcolor{blue}{7} & \textcolor{blue}{8} & \textcolor{blue}{9} \\
        \textcolor{ForestGreen}{11} & \textcolor{ForestGreen}{12} & \textcolor{ForestGreen}{13} & \textcolor{ForestGreen}{14} & \textcolor{ForestGreen}{10} & \textcolor{BurntOrange}{16} & \textcolor{BurntOrange}{17} & \textcolor{BurntOrange}{18} & \textcolor{BurntOrange}{19} & \textcolor{BurntOrange}{15} & \textcolor{red}{1} & \textcolor{red}{2} & \textcolor{red}{3} & \textcolor{red}{4} & \textcolor{red}{0} & \textcolor{blue}{6} & \textcolor{blue}{7} & \textcolor{blue}{8} & \textcolor{blue}{9} & \textcolor{blue}{5} \\
        \textcolor{ForestGreen}{12} & \textcolor{ForestGreen}{13} & \textcolor{ForestGreen}{14} & \textcolor{ForestGreen}{10} & \textcolor{ForestGreen}{11} & \textcolor{BurntOrange}{17} & \textcolor{BurntOrange}{18} & \textcolor{BurntOrange}{19} & \textcolor{BurntOrange}{15} & \textcolor{BurntOrange}{16} & \textcolor{red}{2} & \textcolor{red}{3} & \textcolor{red}{4} & \textcolor{red}{0} & \textcolor{red}{1} & \textcolor{blue}{7} & \textcolor{blue}{8} & \textcolor{blue}{9} & \textcolor{blue}{5} & \textcolor{blue}{6} \\
        \textcolor{ForestGreen}{13} & \textcolor{ForestGreen}{14} & \textcolor{ForestGreen}{10} & \textcolor{ForestGreen}{11} & \textcolor{ForestGreen}{12} & \textcolor{BurntOrange}{18} & \textcolor{BurntOrange}{19} & \textcolor{BurntOrange}{15} & \textcolor{BurntOrange}{16} & \textcolor{BurntOrange}{17} & \textcolor{red}{3} & \textcolor{red}{4} & \textcolor{red}{0} & \textcolor{red}{1} & \textcolor{red}{2} & \textcolor{blue}{8} & \textcolor{blue}{9} & \textcolor{blue}{5} & \textcolor{blue}{6} & \textcolor{blue}{7} \\
        \textcolor{ForestGreen}{14} & \textcolor{ForestGreen}{10} & \textcolor{ForestGreen}{11} & \textcolor{ForestGreen}{12} & \textcolor{ForestGreen}{13} & \textcolor{BurntOrange}{19} & \textcolor{BurntOrange}{15} & \textcolor{BurntOrange}{16} & \textcolor{BurntOrange}{17} & \textcolor{BurntOrange}{18} & \textcolor{red}{4} & \textcolor{red}{0} & \textcolor{red}{1} & \textcolor{red}{2} &  \textcolor{red}{3} & \textcolor{blue}{9} & \textcolor{blue}{5} & \textcolor{blue}{6} & \textcolor{blue}{7} & \textcolor{blue}{8} \\

        \textcolor{BurntOrange}{15} & \textcolor{BurntOrange}{16} & \textcolor{BurntOrange}{17} & \textcolor{BurntOrange}{18} & \textcolor{BurntOrange}{19} & \textcolor{red}{0} & \textcolor{red}{1} & \textcolor{red}{2} & \textcolor{red}{3} & \textcolor{red}{4} & \textcolor{blue}{5} & \textcolor{blue}{6} & \textcolor{blue}{7} & \textcolor{blue}{8} & \textcolor{blue}{9} & \textcolor{ForestGreen}{10} & \textcolor{ForestGreen}{11} & \textcolor{ForestGreen}{12} & \textcolor{ForestGreen}{13} & \textcolor{ForestGreen}{14} \\
        \textcolor{BurntOrange}{16} & \textcolor{BurntOrange}{17} & \textcolor{BurntOrange}{18} & \textcolor{BurntOrange}{19} & \textcolor{BurntOrange}{15} & \textcolor{red}{1} & \textcolor{red}{2} & \textcolor{red}{3} & \textcolor{red}{4} & \textcolor{red}{0} & \textcolor{blue}{6} & \textcolor{blue}{7} & \textcolor{blue}{8} & \textcolor{blue}{9} & \textcolor{blue}{5} & \textcolor{ForestGreen}{11} & \textcolor{ForestGreen}{12} & \textcolor{ForestGreen}{13} & \textcolor{ForestGreen}{14} & \textcolor{ForestGreen}{10} \\
        \textcolor{BurntOrange}{17} & \textcolor{BurntOrange}{18} & \textcolor{BurntOrange}{19} & \textcolor{BurntOrange}{15} & \textcolor{BurntOrange}{16} & \textcolor{red}{2} & \textcolor{red}{3} & \textcolor{red}{4} & \textcolor{red}{0} & \textcolor{red}{1} & \textcolor{blue}{7} & \textcolor{blue}{8} & \textcolor{blue}{9} & \textcolor{blue}{5} & \textcolor{blue}{6} & \textcolor{ForestGreen}{12} & \textcolor{ForestGreen}{13} & \textcolor{ForestGreen}{14} & \textcolor{ForestGreen}{10} & \textcolor{ForestGreen}{11} \\
        \textcolor{BurntOrange}{18} & \textcolor{BurntOrange}{19} & \textcolor{BurntOrange}{15} & \textcolor{BurntOrange}{16} & \textcolor{BurntOrange}{17} & \textcolor{red}{3} & \textcolor{red}{4} & \textcolor{red}{0} & \textcolor{red}{1} & \textcolor{red}{2} & \textcolor{blue}{8} & \textcolor{blue}{9} & \textcolor{blue}{5} & \textcolor{blue}{6} & \textcolor{blue}{7} & \textcolor{ForestGreen}{13} & \textcolor{ForestGreen}{14} & \textcolor{ForestGreen}{10} & \textcolor{ForestGreen}{11} & \textcolor{ForestGreen}{12} \\
        \textcolor{BurntOrange}{19} & \textcolor{BurntOrange}{15} & \textcolor{BurntOrange}{16} & \textcolor{BurntOrange}{17} & \textcolor{BurntOrange}{18} & \textcolor{red}{4} & \textcolor{red}{0} & \textcolor{red}{1} & \textcolor{red}{2} &  \textcolor{red}{3} & \textcolor{blue}{9} & \textcolor{blue}{5} & \textcolor{blue}{6} & \textcolor{blue}{7} & \textcolor{blue}{8} & \textcolor{ForestGreen}{14} & \textcolor{ForestGreen}{10} & \textcolor{ForestGreen}{11} & \textcolor{ForestGreen}{12} & \textcolor{ForestGreen}{13} \\
    \end{tabular}
    \caption{In this example, $k=3$ and $\ell=5$. For $(p,q) \in [k+1] \times [\ell]$, $v_i$ denotes the voter $(p,q)$ with $i=\ell p + q$. For $(x,y) \in [k+1] \times [\ell]$ the candidate $(x,y)$ is labeled as $i = \ell x + y$. The ranking goes from the top (highest ranked) to the bottom (lowest ranked). This instance is adapted from \cite{sixCandidates}.}
    \label{tab:my_label}
\end{table}

We claim that for any committee $C$ of size $k$, there exists a candidate $a \in A \setminus C$ such that at least a $\frac{1+t}{1+k}\left(1-\frac{1}{\ell}\right)$ fraction of voters $v \in [k+1] \times [\ell]$ is such that $a \succ^t_v C$.

We use $s(x)$ to denote the number of candidates in $C$ with the form $(x,\cdot)$. By the fact that $|C| = k$, we have that $\sum_{x=0}^k s(x) = k$. The following states that there exists a small set of neighboring candidates (with the form from $(x-t,\cdot)$ to $(x,\cdot)$) in $C$.

    \begin{claim} \label{cl:pig}
        For any committee $C$ of size $k$, there exists $x \in [k+1]$ such that either
        \begin{enumerate}[label=(\roman*)]
            \item \label{pig-1} $s(x-t) + s(x-t+1) + ... + s(x-1) + s(x) < t$ or
            \item \label{pig-2} $s(x-t) + s(x-t+1) + ... + s(x-1) + s(x) = t$ and $s(x) > 0$.
        \end{enumerate}
    \end{claim}

    \begin{proof}
        For the sake of contradiction, suppose there is no such $x$. This assumption implies that for every $x \in [k+1]$, either
        \begin{itemize}
            \item $s(x) > 0$ and $\sum_{i=0}^t s(x-i) \ge t+1$, or
            \item $s(x) = 0$ and $\sum_{i=0}^t s(x-i) \ge t$.
        \end{itemize}
        From $\sum_{x=0}^k s(x) = k$, we have that
        \begin{align*}
            k(t+1) &= \sum_{x=0}^k \sum_{i=0}^t s(x-i) \\
            &= \sum_{x \in [k+1] : s(x) > 0} \sum_{i=0}^t s(x-i) + \sum_{x \in [k+1] : s(x) = 0} \sum_{i=0}^t s(x-i) \\
            &\ge t+1 + kt
        \end{align*}
        which leads to a contradiction. The last inequality follows from the fact that there must be at least one $x \in [k+1]$ such that $s(x) > 0$.
    \end{proof}
    From Claim \ref{cl:pig}, we do a case analysis on $x \in [k+1]$ that satisfies either \ref{pig-1} or \ref{pig-2}.
    
    If case \ref{pig-1} occurs, then for voters $v$ in the form $(x-t,\cdot),(x-t+1,\cdot), ..., (x-1,\cdot), (x,\cdot)$, any candidate $a = (x,\cdot) \in A \setminus C$ is such that $a \succ^t_v C$. There is a $\frac{1+t}{1+k}$ fraction of such voters.

    If case \ref{pig-2} occurs, then we pick a candidate $a = (x,y-1)$ such that $(x,y) \in C$. Among the $\ell$ voters in the form $(x,\cdot)$, we have $\ell-1$ voters preferring $(x,y-1)$ to $(x,y)$. Since $s(x-t) + s(x-t+1) + ... + s(x-1) + s(x) = t$, for each voter form $(x-t,\cdot),(x-t+1,\cdot), ..., (x-1,\cdot)$, and $(x,\cdot)$, there is at least a $1-\frac{1}{\ell}$ fraction of voters $v$ such that $a \succ^t_v C$. 
    Overall, there is at least a $\frac{1+t}{1+k}\left(1-\frac{1}{\ell}\right)$ fraction of such voters.

    To conclude, for any $\alpha < \frac{1+t}{1+k}$, there exists $\ell$ such that $\alpha < \frac{1+t}{1+k}\left(1-\frac{1}{\ell}\right)$. Applying the construction above with such an $\ell$, the result follows.
    \end{proof}

\subsection{Proof of Theorem \ref{thm:sSLEO}} \label{app:pf-thm:sSLEO}

\thmssleo*

\begin{proof}
We use a strategy similar to Appendix \ref{app:LEO}.

Define $\Gamma:= \{\z \in \mathbb{R}_+^{|A|} \mid \mathbf{1}^T\z = B \text{ and } \z \le \mathbf{1}\} $. We construct the following correspondence:


$$\mathcal{L}: \prod_{v=1}^n[0,1]^{|A \cup \{\emptyset\}|}\times \Gamma \times \prod_{v=1}^n[0,1]^{|A \cup \{\emptyset\}|} \rightrightarrows \prod_{v=1}^n[0,1]^{|A \cup \{\emptyset\}|} \times \Gamma  \times\prod_{v=1}^n[0,1]^{|A \cup \{\emptyset\}|}$$
with
$$\mathcal{L}\left((\{\x_{v}\}_{v\in V},\y,\{\prices_{v}\}_{v\in V})\right) = (\{\x'_{v}\}_{v\in V},Y',\{\prices'_{v}\}_{v\in V})$$ where
\begin{subequations}
\begin{align}
x'_{v,a} &= \Pr[\mathcal{D}_{v}(\prices_{v}, \randbud{}) =a] & \forall v \in V, a \in A \cup \{\emptyset\}, \label{eq:sfp-x} \\
Y' &= \arg \max_{\z \in \Gamma} \;\; \left(\sum_{v\in V} \prices_v\right)^T \z, \label{eq:sfp-y} \\ 
p'_{v,a} &= 
\begin{cases}
\max\left\{\min\left\{1,p_{v,a}+(s \cdot x'_{v,a}-y_a)\right\} ,0 \right\} \text{  if } a \in A,\\
0 \text{  if } a = \emptyset,
\end{cases} & \forall v \in V, a \in A \cup \{\emptyset\}. \label{eq:sfp-p}
\end{align}
\end{subequations}
By analogous reasoning from Appendix \ref{app:LEO}, we have that
\begin{enumerate}[label=(\roman*)]
    \item \label{pf-ssleo-1} By Theorem \ref{thm:maximum}, the correspondence $\prod_{v=1}^n[0,1]^{|A \cup \{\emptyset\}|} \rightrightarrows \Gamma$ (from $\prices$ to $Y'$) is upper hemi-continuous, non-empty, and compact valued. Furthermore, $Y'$ is convex, bounded, and non-empty. These statements still hold because $\Gamma$ is compact and $Y'$ is the set of optimal solutions of a linear program with a bounded and non-empty feasible region as $B \le |A|$ and $\Gamma:= \{\z \in \mathbb{R}_+^{|A|} \mid \mathbf{1}^T\z = B \text{ and } \z \le \mathbf{1}\} $.
    \item By Theorem \ref{thm:kakutani}, there exists a \emph{fixed point} $(\{\x_{v}\}_{v\in V},\y,\{\prices_{v}\}_{v\in V})$ such that 
    $$(\{\x_{v}\}_{v\in V},\y,\{\prices_{v}\}_{v\in V}) \in \mathcal{L}(\{\x_{v}\}_{v\in V},\y,\{\prices_{v}\}_{v\in V}).$$
    This holds since $\mathcal{L}$ is upper hemi-continuous, following the same reasoning that the mappings from $(\{\x_{v}\}_{v\in V},\y,\{\prices_{v}\}_{v\in V})$ to $\x'$ and $\prices'$ are continuous, and to $Y'$ is upper hemi-continuous because of \ref{pf-ssleo-1}.
    \item The fixed point $(\{\x_{v}\}_{v\in V},\y,\{\prices_{v}\}_{v\in V})$ satisfies all the conditions of a LEO in Definition \ref{def:sSLEO}. This holds because conditions \ref{ssleo-1} and \ref{ssleo-3} follow directly by \eqref{eq:sfp-x} and \eqref{eq:sfp-p}, respectively. Condition \ref{ssleo-2} follows by the same reasoning in Appendix \ref{app:LEO}. The main difference is forcing $s \cdot x_{v,a} \le y_a$, and whenever $s \cdot x_{v,a} < y_a$, we must have $p_{v,a} = 0$.
\end{enumerate}
\end{proof}

\subsection{Proof of Claim \ref{clm:algorithm-correctness}} \label{app:pf-clm:algorithm-correctness}

\clmalgorithmcorrectness*

\begin{proof}

We first prove the correctness of Algorithm \ref{alg:iterative}, i.e., the output $C$ of the iterative algorithm is indeed a $(t, \frac{\alpha}{1-\epsilon})$-undominance set. To prove this, we notice that during the process, adding more members will never make any voter worse off. 

Let $i^*$ be the last index for which $B_i\geq \gamma t$ so that $B_{i^*+1}\leq \gamma t$. Then, assuming the worst case scenario in which all voters in $V_{i^*+1}$ will deviate to $a$, we have 
$$ \big | v\in V: a\succ_v^t C\big| \leq \bigg(\sum_{i=1}^{i^*}\big|v\in S(C_i,V_i): a\succ_v^t C_i \big|\bigg) + \big|V_{i^*+1}\big |.$$

As seen in the proof of Claim \ref{lemma:geometric}, we can bound $|V_{i^*+1}|$ as: 
\begin{align*}
|V_i^* + 1 | & \leq |V|\cdot \omega(\gamma,\tau)^{i^*+1}\\
& \leq \frac{B_0}{\tau^{i^*+1}}\cdot \frac{|V|}{B_0}\cdot [\omega(\gamma,\tau)\cdot \tau]^{i^*+1}\\
& \leq B_{i^*+1}\cdot \frac{|V|}{B_0}\cdot [\omega(\gamma,\tau)\cdot \tau]^{i^*+1}\\
& \leq \gamma t\cdot \frac{|V|}{B_0}\cdot [\omega(\gamma,\tau)\cot \tau]^{i^*+1}.
\end{align*}

Using Claim \ref{claim:deviating_people} and plugging in the value of $B_0$, we can obtain the following.  
\begin{align*}
\big | v\in V: a\succ_v^t C\big| 
& \leq \bigg(\sum_{i=1}^{i^*}\big|v\in S(C_i,V_i): a\succ_v^t C_i \big|\bigg) + \big|V_{i^*+1}\big |\\
& \leq  \sum_{i = 1}^{i^*+1} \frac{\gamma t}{1-\epsilon}\cdot [\tau\omega(\gamma,t)]^i\cdot \frac{|V|}{B_0}\\
& \leq \sum_{i = 1}^{\infty}\frac{\gamma t}{1-\epsilon}\cdot [\tau\omega(\gamma,t)]^i\cdot \frac{|V|}{B_0}\\
& \leq \frac{\gamma t}{1-\epsilon}\cdot \frac{|V|}{B_0} \cdot \sum_{i} [\tau\omega(\gamma,t)]^i\\
& \leq \frac{\gamma t}{1-\epsilon}\cdot \frac{|V|}{B_0} \cdot \frac{1}{1-\tau\omega(\gamma,t)}     \\
& \leq \frac{1}{1-\epsilon}\cdot \alpha n.
\end{align*}

Now we want to bound the size of $C$. Since we require $B_i\geq \gamma t$ during the execution of the algorithm, $|C_i|\leq \lceil B_i\rceil \leq B_i+1$. So, we obtain

\begin{align*}
|C|\leq \sum_{i: \gamma t \leq B_i \leq B_0}|C_i|\leq \sum_{i: \gamma t \leq B_i \leq B_0} (B_i+1) &= B_0 \sum_{i: \gamma t \leq B_i \leq B_0}\frac{1}{\tau^i}+\sum_{i: \gamma t \leq B_i \leq B_0} 1\\  
& \leq B_0\cdot \sum_{i=0}^{\infty} \frac{1}{\tau^i} + \log_{\tau}(\frac{B_0}{\gamma t})\\
&=  \frac{\gamma \tau}{\tau - 1 }\cdot \frac{1}{1-\omega(\gamma,t)\cdot \tau}\cdot \frac{t}{\alpha}+ \log_{\tau}(\frac{1}{1-\omega(\gamma, t)\cdot \tau}\cdot \frac{1}{\alpha})
\end{align*}
and this completes the proof.
\end{proof}

\subsection{Proof of Theorem \ref{thm:generalized-cond} (Omitted Part)} \label{app:pf-thm:generalized-cond}

\thmgeneralupperbound*

\begin{proof}

Now we show that for any integer $t\geq 2$, $\delta(t)\leq 4.75$. We can first verify that whenever $\alpha>0.045$, $s_2(\alpha, t)\leq 4.75\cdot \frac{t}{\alpha}$ for all $t \ge 2$, as shown in Figure \ref{fig:sawtooth}. 

So to prove $\delta(t)\le 4.75$ for $t\ge 2$, we assume $\alpha\leq 0.045$ and only consider $s_1$. Then we can verify that for any $\alpha\leq 0.045$, $\tau \geq 3.47$, $\gamma\geq 1$, and $t \ge 2$, we have 
$$\log_\tau(\frac{1}{1-\omega(\gamma, t)\cdot \tau}\cdot \frac{1}{\alpha})\leq 0.025\cdot \frac{\gamma\cdot \tau}{\tau-1}\cdot \frac{1}{1- \omega(\gamma,t)\tau}\cdot \frac{t}{\alpha} .$$
With this, we can compute that whenever $\alpha \le 0.045$ and $t \ge 2$: 
\begin{align*}
s_1(\alpha, t)& \leq  \min_{\gamma\geq 1, \tau\geq 3.47} \left\{\frac{\gamma \tau}{\tau-1}\cdot \frac{1}{1- \omega(\gamma,t)\tau}\cdot \frac{t}{\alpha} + \log_{\tau}(\frac{1}{1-\omega(\gamma, t) \tau}\cdot \frac{1}{\alpha})\right\}\\
& \leq 1.025\cdot  \min_{\gamma\geq 1, \tau\geq 3.47} \left\{\frac{\gamma \tau}{\tau-1}\cdot \frac{1}{1- \omega(\gamma,t)\tau}\cdot \frac{t}
{\alpha}\right\}\\
& \leq 1.025\cdot \frac{t}{\alpha}\cdot  \min_{\gamma\geq 1, \tau\geq 3.47} \left\{\frac{\gamma \tau}{\tau-1}\cdot \frac{1}{1- \omega(\gamma,t)\tau}\right\}\\
&\leq 4.75 \cdot \frac{t}{\alpha}.
\end{align*}

Then we need to show that $\delta(t)\rightarrow 1$ as $t\rightarrow \infty$. It suffices to show $s_2(\alpha, t)\rightarrow \frac{t}{\alpha}$ as $t\rightarrow \infty$. For any $\eta>0$, we need to find a $T\in \mathbb{N}$ s.t. for any $t\geq T$, $$s_2(\alpha, t)\leq \left\lceil\frac{(1+\eta) t}{\alpha}\right\rceil.$$
We first pick a $\gamma_0 = 1+\frac{\eta}{2}$ and for each $t$, let $B_t =  \left\lceil\frac{(1+\eta) t}{\alpha}\right\rceil$. Since $\omega(\gamma_0, t)\rightarrow 0$ as $t\rightarrow \infty$, we can find a $T$ s.t. for any $t\geq T$, 
$$\omega(\gamma_0,t) \le \alpha \left(1 - \frac{\gamma_0}{1+\eta}\right) \implies \frac{\alpha \gamma_0}{(1+\eta)} + \omega(\gamma_0,t) \le \alpha \implies \frac{\gamma_0 t}{B_t} + \omega(\gamma_0, t) \leq \alpha.$$
Consequently, we have for any $t\geq T$, 
$$ s_2(\alpha ,t )\leq B_t =\left\lceil\frac{(1+\eta) t}{\alpha}\right\rceil$$
and this completes the proof. 
\end{proof}

\newpage

\section{Plots for $s_1$ and $s_2$ in the Proof of Theorem \ref{thm:generalized-cond}} \label{app:plot}

\input{t_alpha_und}

\newpage

\bibliographystyle{ACM-Reference-Format}
\bibliography{reference}

\end{document}